\documentclass[pra,aps,twocolumn,nofootinbib,superscriptaddress,longbibliography,10pt]{revtex4-2}

\usepackage{algorithm2e}
\usepackage{amsfonts}
\usepackage{amsmath}
\usepackage{amssymb}
\usepackage{mathtools}
\usepackage[colorinlistoftodos]{todonotes}
\usepackage{array}
\usepackage{caption}
\usepackage{subcaption}
\usepackage{graphicx}
\usepackage[ colorlinks = true, linkcolor = blue, urlcolor=blue, citecolor = red, anchorcolor = green,
]{hyperref}
\providecommand{\U}[1]{\protect\rule{.1in}{.1in}}
\newtheorem{theorem}{Theorem}

\newtheorem{Algorithm}{Algorithm}

\newtheorem{definition}{Definition}

\newtheorem{lemma}[theorem]{Lemma}

\newtheorem{remark}{Remark}

\newenvironment{proof}[1][Proof]{\noindent\textbf{#1.} }{\ \rule{0.5em}{0.5em}}
\numberwithin{theorem}{section}
\numberwithin{proposition}{section}
\numberwithin{definition}{section}
\numberwithin{example}{section}
\usepackage{blindtext}

\newcommand{\bra}[1]{\langle#1|}
\newcommand{\ket}[1]{|#1\rangle}

\newcommand{\outerproj}[1]{\vert #1 \rangle\!\langle #1 \vert}

\newcommand{\dopr}[2]{\vec{#1} \cdot \vec{#2}}

\begin{document}

\title[]{Logarithmic-Depth Quantum Circuits for Hamming Weight Projections}
\author{Soorya Rethinasamy}
\affiliation{School of Applied and Engineering Physics, Cornell University, Ithaca, New
York 14850, USA}
\author{Margarite L. LaBorde}
\affiliation{Hearne Institute for Theoretical Physics and Department of Physics and Astronomy, Louisiana State University, Baton Rouge, Louisiana 70803, USA}
\author{Mark M. Wilde}
\affiliation{School of Electrical and Computer Engineering, Cornell University, Ithaca, New
York 14850, USA}

\begin{abstract}

A pure state of fixed Hamming weight is a superposition of computational basis states such that each bitstring in the superposition has the same number of ones. Given a Hilbert space of the form~$\mathcal{H} = (\mathbb{C}_2)^{\otimes n}$, or an $n$-qubit system, the identity operator can be decomposed as a sum of projectors onto subspaces of fixed Hamming weight. In this work, we propose several quantum algorithms that realize a coherent Hamming weight projective measurement on an input pure state, meaning that the post-measurement state of the algorithm is the projection of the input state onto the corresponding subspace of fixed Hamming weight. We analyze a depth-width trade-off for the corresponding quantum circuits, allowing for a depth reduction of the circuits at the cost of more control qubits. For an $n$-qubit input, the depth-optimal algorithm uses $\mathcal{O}(n )$ control qubits and the corresponding circuit has depth $\mathcal{O}(\log (n))$, assuming that we have the ability to perform qubit resets. Furthermore, the proposed algorithm construction uses only one- and two-qubit gates.
\end{abstract}

\date{\today}
\startpage{1}
\endpage{10}
\maketitle

\tableofcontents

\section{Introduction}

Quantum computing is a computational paradigm that, for particular computational tasks, can potentially provide speedups to the best-known classical algorithms. Multiple classes of problems have now been shown to have quantum algorithms solving them faster than the current fastest classical algorithms \cite{dalzell2023quantumalgorithmssurveyapplications}. To leverage the power of quantum computing, several key properties, like entanglement and superposition, need to be manipulated effectively. 


The Hamming distance of two binary strings is the number of substitutions needed to convert one string to the other. A special case of the Hamming distance is the Hamming weight, which is the Hamming distance of a string from the all-zeros string. The Hamming weight is more simply defined as the number of `1' bits in the string, or, equivalently, the $\ell_1$-norm of a binary vector~\cite{cusick2017cryptographic}. The Hamming weight of a string is a relevant quantity in multiple different fields such as cryptography~\cite{Kangquan2022generalized},  error correction~\cite{xu2023subfield}, and information theory~\cite{CTbook}. 
Many classical algorithms determine Hamming weight to facilitate applications such as internet-application lookup protocols~\cite{Stoica2003Chord} or determining the cryptographic performance of linear codes~\cite{wei1991generalized}. An entire suite of functions, called symmetric Boolean functions, rely solely on the Hamming weight of the input, encompassing foundational logical tasks such as AND, OR, Majority, and Parity~\cite{Wegener1987}.  Indeed, since classical computation primarily uses binary operations to perform calculations, the Hamming weight arises naturally in a variety of situations. 

A quantum generalization of the Hamming weight is straightforwardly defined on computational basis strings; however, in the case of superpositions, the Hamming weight of the overall state is not well defined. In such a scenario, a Hamming weight measurement can have different outcomes depending on the specific superposition. For example, a state of the form 
\begin{equation}
    \ket{\psi} = \frac{1}{2}\left( \ket{000} + \ket{001} + \ket{010} + \ket{111}\right),
\end{equation}
has a probability of $1/4$ to be observed in the Hamming subspace $0$, a probability of $1/2$ for the Hamming subspace $1$, and a probability of $1/4$ for the Hamming subspace $3$.

The main contribution of our paper is a fast quantum algorithm for performing a coherent Hamming weight measurement, as a strongly related sequel to our Hamming-weight symmetry tests from~\cite[Sections~6.4.1 and 6.4.2]{LaBorde2023testingsymmetry}.\footnote{Specifically, the following was noted after~\cite[Eq.~(189)]{LaBorde2023testingsymmetry}: ``As an aside, we note that a generalization
of our method allows for performing a projection onto
constant-Hamming-weight subspaces, which is useful in
tasks like entanglement concentration.''} ``Coherent'' in this context means that, not only is the classical measurement outcome equal to the Hamming weight, but it is also required for coherences in the same Hamming weight subspace to be preserved.
In more detail, the aim of the coherent quantum Hamming weight measurement is to realize the following quantum channel on an $n$-qubit state $\rho$:
\begin{equation}
    \rho \mapsto \sum_x |x\rangle\!\langle x| \otimes \Pi_x \rho \Pi_x,
\end{equation}
where $x$ is the classical measurement outcome, equal to the Hamming weight, and $\Pi_x$ is the projection onto the subspace of $\mathbb{C}_2^{\otimes n}$ with Hamming weight $x$. See \eqref{eq:Px_def} for an explicit example. The probability to observe the outcome~$x$ is
\begin{equation}
p(x) = \operatorname{Tr}[\Pi_x \rho],
\end{equation}
and the post-measurement state, conditioned on this outcome, is given by 
\begin{equation}
\frac{\Pi_x \rho \Pi_x }{p(x)}.
\label{eq:post-meas-state-Hamm}
\end{equation}

The essential aspect of the coherent Hamming weight measurement is the post-measurement state. Contrary to what our algorithm accomplishes, an incoherent Hamming weight measurement, which has outcome probability distribution $(p(x))_x$ without concern for the post-measurement state, can be realized by simply measuring all qubits in the computational basis and calculating the Hamming weight classically by means of known fast classical algorithms~\cite{warren2012hacker}. The coherent version that leads to the post-measurement state in \eqref{eq:post-meas-state-Hamm} allows our protocol to be used as a subroutine in other quantum information-processing tasks. 

Prior work on the coherent Hamming weight measurement includes~\cite{KM2001}, which provides a linear-depth quantum circuit for this task. Interestingly, this circuit was considered experimentally in an ion-trap setup in~\cite{SLLSH16}. Another work considered a coherent Hamming weight measurement as a collective magnetization measurement~\cite{MCE18,Mirkamali2019}, but there is no discussion of the circuit depth. In contrast to~\cite{KM2001}, our circuit has logarithmic depth and requires a linear number of control qubits, assuming that qubit resets are allowed. 

One of the key components of our algorithm, that allows for the reduction in depth, is the quantum fan-out operation, implemented in constant quantum-depth using mid-circuit measurements. The usage of the quantum fan-out is known to help in various other algorithms~\cite{HS05}, including factoring using Shor’s algorithm~\cite{2D_Factoring}. Another component of our algorithm involves using controlled rotation gates to calculate Hamming weights coherently, which bears similarities to the earlier work of~\cite{draper2000}.

The ability to perform a coherent Hamming weight measurement has various applications in quantum information science. Realizing a coherent Hamming weight measurement efficiently is useful in information processing tasks like entanglement concentration~\cite{BBPS96} (see also~\cite[Chapter~19]{Wilde_2017}), multiparty entanglement concentration~\cite{SW22}, coherence concentration~\cite{PhysRevLett.116.120404}, multicasting of pure quantum states~\cite{SHIH2010273}, and compression of pure states~\cite{PhysRevA.81.032317}. 
Another application of the coherent Hamming weight measurement is in error detection and correction. In the natural setting of independent and identically distributed (i.i.d.)~amplitude damping noise, a coherent Hamming weight measurement can be used to detect which error has occurred when using the constant-excitation codes of~\cite{OUY21}. For error correction of i.i.d.~amplitude damping noise, when using such codes, it is possible to perform the correction step based on the constant-excitation space projected into, by employing a coherent Hamming weight measurement.

In this work, we present an algorithm that uses $  n$ control qubits and has $\mathcal{O}(\log (n))$ depth to realize the coherent Hamming weight measurement on a quantum state of $n$~qubits. To keep the number of control qubits to be~$n$, we assume the ability to perform mid-circuit measurements and qubit resets as the algorithm proceeds. Dynamic circuits on IBM hardware were used in~\cite{PhysRevLett.127.100501} and compiling quantum circuits to use qubit resets was employed in~\cite{PhysRevX.13.041057}.

Figure~\ref{fig:full-circuit-resets} depicts the corresponding quantum circuit.
 We also describe a trade-off between circuit depth and number of control qubits. The other end of the spectrum is an algorithm that uses $\mathcal{O}(\log(n))$ control qubits and has depth $\mathcal{O}(n\log (n))$.

\medskip 
\textit{Note on independent work}---While finalizing our paper, we came across the concurrent and independent works~\cite{piroli2024approximating,zi2024shallow}, which also provide shallow-depth circuits for computing Hamming weight coherently.


\section{Preliminaries}
\label{sec:Prelim}

In this section, we present some preliminary aspects of the algorithm. Let $\ket{\psi}$ be an $n$-qubit input state. For an $n$-qubit quantum state, there are $(n+1)$ possible Hamming weight outcomes: $\{0, 1, \ldots, n\}$. Define the following:
\begin{align}
    k & \coloneqq \lceil \log _2(n+1) \rceil \label{eq:def_K} , \\
    w & \coloneqq 2^k - (n + 1) \label{eq:def_w}.     
\end{align}
Intuitively, if we were to store the $(n+1)$ possible Hamming weight values, we would need at least~$k$ bits. However, we find it convenient for this initial discussion to pad to the next power of two for computational purposes; thus, the number of extra outcomes with~$k$ bits is given by $w$. Later, we eliminate the need for this extra padding. In the case where $(n+1)$ is a power of two, $w=0$. We now define the zero-padded state
\begin{equation}
    \ket{\phi} \coloneqq \ket{\psi}\ket{0}^{\otimes w}.
    \label{eq:zero-padded-state}
\end{equation}
Thus, $\ket{\phi}$ has the same Hamming weight as $\ket{\psi}$ and is on a register of $N \coloneqq 2^k-1$ qubits.

Next, let $P_x$ denote the projector onto the space of fixed Hamming weight $x$. For example, consider a two-qubit state $\ket{\psi}$. Then, we append $w=1$ qubits in the $\ket{0}$ state and set $N=3$. The Hamming weight projectors for $x \in \{0, 1, \ldots, N\}$ are as follows:
\begin{align}
\label{eq:Px_def}
    P_0 &= \outerproj{000} \nonumber ,\\
    P_1 &= \outerproj{001} + \outerproj{010} + \outerproj{100} \nonumber ,\\
    P_2 &= \outerproj{011} + \outerproj{101} + \outerproj{110} \nonumber, \\
    P_3 &= \outerproj{111}.
\end{align}
The Hamming weight projective measurement is thus defined by the measurement operator set $\{P_x\}_x$ for $x \in \{0, \ldots, N\}$.

To realize this measurement, we use a construction based on the idea of a pinching channel. See~\cite[Section~2.6.3]{tomamichel2015quantum} for a review of this concept, as well as the construction mentioned below. A pinching map is a channel of the form 
\begin{equation}
    P \colon L \to \sum_x P_x L P_x
\end{equation}
where $\{P_x\}_x$ for $x \in \{0, \ldots, N \}$ are orthogonal projectors that sum up to the identity. Importantly, a pinching map has an alternative representation as a random unitary channel:
\begin{equation}
    \sum\limits_{x=0}^N P_x L P_x = \frac{1}{N+1} \sum\limits_{y=0}^N U_y L U^\dagger_y,
\end{equation}
where the set $\{U_y\}_y$ for $y \in \{0, \ldots, N\}$ is defined as
\begin{equation}
    \label{eq:Uy_initial_def}
    U_y \coloneqq \sum\limits_{x=0}^N \exp\!\left(\frac{2 \pi \mathrm{i} y x}{N+1} \right) P_x ,
\end{equation}
and the equality results from basic Fourier analysis. The set $\{U_y\}_y$  is a unitary representation of the cyclic group~$C_{N+1}$. To see this, we first define a cyclic group on $M$ elements:
\begin{equation}
    C_M \coloneqq \langle a \mid a^M = e \rangle,
\end{equation}
where $e$ is the identity element and $a$ is the generator of the group. To show that the set $\{U_y\}_y$ is a unitary representation of $C_{N+1}$, we first note that
\begin{equation}
    U_0 = \sum\limits_{x=0}^N P_x = \mathbb{I}.
\end{equation}
Thus, $U_0$ is the unitary representation of $e$. Next, we show that $U_1^{N+1} = \mathbb{I}$. To see this, we expand
\begin{equation}
    U_1^{N+1} = \left[\sum\limits_{x=0}^N \exp\!\left(\frac{2 \pi \mathrm{i} x}{N+1} \right) P_x \right]^{N+1}.
\end{equation}
Since the set $\{P_i\}_i$ consists of  mutually orthogonal projectors, this evaluates to
\begin{align}
    U_1^{N+1} &= \sum\limits_{x=0}^N \exp(2 \pi \mathrm{i} x) P_x \nonumber \\
    &= \mathbb{I}.
\end{align}
Thus, the set $\{U_y\}_y$ is a unitary representation of the group $C_{N+1}$.

The next order of business is to find a way to implement the set of unitaries $\{U_y\}_y$ on a quantum computer. For the case of Hamming-weight projections, each of the unitaries can be written as a product of $Z$ rotations on individual qubits, as shown in~\cite[Section~6.4.1]{LaBorde2023testingsymmetry}. Here we recall these points briefly. To this end, first consider the representation of the operator $R_z(\phi)$ in the computational basis:
\begin{equation}
    R_z(\phi) = \operatorname{Diag}\left\{\exp\!\left(-\frac{\mathrm{i}\phi}{2} \right), \exp\!\left(\frac{\mathrm{i}\phi}{2} \right)\right\}.
\end{equation}
Similarly, expressing $R_z(\phi)^{\otimes 2}$ in the computational basis gives
\begin{equation}
    R_z(\phi)^{\otimes 2} = \operatorname{Diag}\left\{\exp\!\left(-\mathrm{i}\phi\right), 1, 1, \exp\!\left(\mathrm{i}\phi\right)\right\}.
\end{equation}
The~$k$th entry in the expansion depends on the number of zeros and ones in the binary expansion of the number~$k$. To generalize to the case of $N$ qubits, we observe that the number of zeros in a bit-string $x$ is $N-H(x)$ and the number of ones is $H(x)$, where $H(x)$ is the Hamming weight of $x$. For example, $H(6) = 2$ since $6_{10} \equiv 110_2$. Each zero contributes a phase of $-\phi/2$ for a total of $-(N-H(x))\phi/2$, and each one contributes a phase of $\phi/2$,  for a total of $H(x) \phi / 2$. Then the overall total for the bit-string $x$ is 
\begin{equation}
-(N-H(x))\phi/2 + H(x) \phi / 2 = (2H(x) - N)\phi/2.   
\end{equation}
This implies that 
\begin{align}
    R_z(\phi)^{\otimes N} &= \operatorname{Diag}\left\{ \exp \!\left[ \left( \frac{2H(x)-N}{2} \right) \mathrm{i}\phi \right]_{x=0}^{2^N-1} \right\} \nonumber \\
    &= \sum\limits_{x \in \{0, 1\}^N} \exp \!\left[ \left(\frac{2H(x)-N}{2}\right) \mathrm{i}\phi \right] \outerproj{x} \nonumber \\
    &= \exp\! \left(\frac{-\mathrm{i}N\phi}{2} \right) \sum\limits_{x \in \{0, 1\}^N} \exp \!\left[ \mathrm{i}H(x)\phi \right] \outerproj{x} \notag \\
    &= \exp\! \left(\frac{-\mathrm{i}N\phi}{2} \right) \sum_{x=0}^N \exp \!\left[ \mathrm{i}x\phi \right] P_x,
    \label{eq:Rz-phi-otimes-N}
\end{align}
where the last equality results from grouping together all the projectors of fixed Hamming weight. We now make a particular choice of the angle $\phi$, as follows:
\begin{equation}
    \phi = \frac{2\pi y}{N+1}.
\end{equation}
Thus,
\begin{equation}
    R_z\!\left(\frac{2\pi y}{N+1}\right)^{\otimes N} = \exp\! \left(\frac{-\mathrm{i} \pi y N }{N+1} \right)\\ \sum_{x=0}^N \exp \!\left[\frac{2\pi \mathrm{i} y x}{N+1} \right] P_x.
\end{equation}
Comparing with \eqref{eq:Uy_initial_def}, we see that
\begin{equation}
    \label{eq:Uy_new_def}
    U_y = \exp\! \left(\frac{\mathrm{i} \pi y N }{N+1} \right) R_z\!\left(\frac{2\pi y}{N+1}\right)^{\otimes N}.
\end{equation}
Thus, each unitary $U_y$ is composed of an overall phase and a product of $Z$ rotations. The group structure of $\{U_y\}_y$ is more evident in this form, and the unitaries $U_y$ can be realized on a quantum computer.

We now give some remarks and lemmas that will be useful for the proofs in the next section.

\begin{remark}
\label{rmk:Rz_computational}
The action of $R_z(\phi)$ on a computational basis state can be written as
\begin{equation}
    R_z(\phi)\ket{x} = \exp{\!\left(\frac{\mathrm{i}\phi}{2}(2x-1)\right)} \ket{x},
\end{equation}
for $x \in \{0, 1\}$.
\end{remark}

\begin{remark}
\label{rmk:C-Rz_computational}
The action of a controlled-$R_z(\phi)$ (abbreviated hereafter $CR_z(\phi)$) when the control qubit is in a computational basis state can be written as
\begin{equation}
    CR_z(\phi)_{AB}\ket{x}_A\ket{\psi}_B = \ket{x}_A R_z(x\phi) \ket{\psi},
\end{equation}
for $x \in \{0, 1\}$.
\end{remark}

\begin{lemma}[Reduction of Controlled Rotations]
\label{lem:Rz-cancellation}
Consider the following two-qubit state:
\begin{equation}
    \left[\alpha \ket{0}_C + \beta\ket{1}_C\right] \ket{0}_S.
\end{equation}
The action of a controlled-$R_z$ rotation on a state of this form can be replaced with a corresponding $R_z$ rotation on the control qubit. Indeed, this result can be seen from the following reasoning:
\begin{align}
    &CR_z(\phi)\left[\alpha \ket{0}_C + \beta\ket{1}_C\right] \ket{0}_S \notag \\
    &= \left(\outerproj{0}_C \otimes I_S + \outerproj{1}_C \otimes R_z(\phi)_S\right) \times \notag \\
    &\qquad\qquad\qquad\qquad \left[\alpha \ket{00}_{CS} + \beta\ket{10}_{CS}\right] \\
    &= \alpha \ket{00}_{CS} + R_z(\phi)_S \beta\ket{10}_{CS} \\
    &= \left[\alpha \ket{0}_C + e^{\frac{-\mathrm{i} \phi}{2}} \beta\ket{1}_C\right] \ket{0}_S \\
    &= e^{\frac{-\mathrm{i} \phi}{4}} \left[e^{\frac{\mathrm{i} \phi}{4}} \alpha \ket{0}_C + e^{\frac{-\mathrm{i} \phi}{4}} \beta\ket{1}_C\right] \ket{0}_S \\
    &= e^{\frac{-\mathrm{i} \phi}{4}} R_z\!\left(\frac{-\phi}{2}\right)_C \left[\alpha \ket{0}_C + \beta\ket{1}_C\right] \ket{0}_S.
    \label{eq:removeC-Rz}
\end{align}
Due to the leftmost phase $e^{\frac{-\mathrm{i} \phi}{4}}$ being a global phase, it has no physically observable consequence: the state of the target qubit is unchanged, and the overall effect is to rotate the control qubit. Thus, a controlled-$R_z(\phi)$ gate acting on a state where the target qubit is in the state $\ket{0}$ can be replaced with $R_z(\frac{-\phi}{2})$ acting on the control qubit. 
\end{lemma}

\section{Algorithm Construction}

\label{sec:AlgCons}
In this section, we put forth two algorithms to realize the coherent Hamming weight measurement on an $n$-qubit input state $\ket{\psi}$. Recall the definitions of $k$ and $w$ from \eqref{eq:def_K} and \eqref{eq:def_w}, respectively. 

The first algorithm uses zero padding to have exactly $2^k$ measurement outcomes. This simplifies the analysis and the proof, but it requires more ancilla qubits and rotation gates. The second algorithm eliminates the extra zero-padded ancilla qubits and rotation gates. The analysis of this algorithm builds upon the analysis of the first algorithm.

\subsection{First Algorithm with Zero Padding}

The first algorithm uses the zero-padded state $\ket{\phi}$ of  $N = 2^k - 1$ qubits, as defined in~\eqref{eq:zero-padded-state}, and it is depicted in Figure~\ref{fig:Circuit_Uy}. It has $\mathcal{O}(\log n)$ control qubits.

\begin{Algorithm}
\label{alg:hamming-wt-meas-ineff} 

\begin{figure}
    \includegraphics[width=1.0\columnwidth]{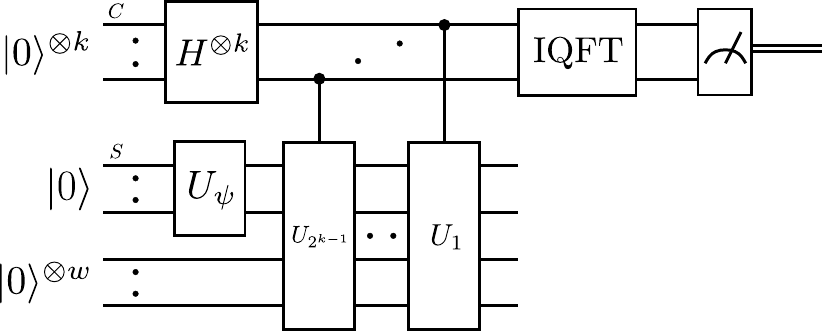}
    \caption{Circuit to realize the quantum Hamming weight of an input state. The unitary $U_{\psi}$ prepares the state~$\ket{\psi}$, on which the coherent Hamming weight measurement is being performed.}
    \label{fig:Circuit_Uy}
\end{figure} 

The algorithm consists of the following steps:
\begin{enumerate}
\item Prepare a control register $C$, consisting of~$k$ qubits, in the state $|0\rangle_{C}$.

\item Act on register $C$ with $k$ Hadamard gates, leading to the state
\begin{align}
    \ket{+}_C &= H^{\otimes k}\ket{0}_C^{\otimes k} \nonumber \\
    &= \frac{1}{\sqrt{N+1}} \sum\limits_{y=0}^N \ket{y}_C.
\end{align}

\item Append the zero-padded state $\ket{\phi} = \ket{\psi}\ket{0}^{\otimes w}$, and perform the following controlled unitary:
\begin{equation}
\sum\limits_{y=0}^N \outerproj{y}_{C} \otimes U_y,
\end{equation}
where $U_y$ is defined in \eqref{eq:Uy_new_def}.

\item Perform an inverse quantum Fourier transform on register~$C$ and measure in the computational  basis $\{\outerproj{0}, \ldots, \outerproj{N} \}$.
\end{enumerate}

\end{Algorithm}

\begin{theorem}
Algorithm~\ref{alg:hamming-wt-meas-ineff} realizes a coherent Hamming-weight measurement $\{P_x\}_{x=0}^{n}$ on the state $\ket{\psi}$.
\end{theorem}
\begin{proof}
The result of Step~3 of Algorithm~\ref{alg:hamming-wt-meas-ineff} is to prepare the following superposition state:
\begin{equation}
\frac{1}{\sqrt{N+1}} \sum\limits_{y=0}^N \ket{y}_{C} \otimes U_y \ket{\phi}. 
\end{equation}

Upon acting on the control qubits with the inverse quantum Fourier transform, the state becomes:
\begin{equation}
    \frac{1}{N+1} \sum\limits_{y=0}^N \sum\limits_{z=0}^N \exp\!\left(\frac{-2 \pi \mathrm{i} y z}{N+1} \right) \ket{z}_{C} \otimes U_y \ket{\phi}. 
\end{equation}
Expanding the unitary $U_y$ using \eqref{eq:Uy_initial_def} gives
\begin{equation}
    \label{eq:final_state_triple_sum}
    \frac{1}{N+1} \sum\limits_{z,x,y=0}^N \exp\!\left(\frac{2 \pi \mathrm{i} y (x-z)}{N+1} \right) \ket{z}_{C} \otimes P_x \ket{\phi}.
\end{equation}
Consider the following summation:
\begin{equation}
    \label{eq:summation-S-init}
    S = \sum\limits_{y=0}^N \exp\!\left(\frac{2\pi \mathrm{i} y(x-z)}{N+1} \right).
\end{equation}
Define $c = x-z$. For all $c \in \mathbb{Z} \setminus \{0\}$, 
\begin{align}
    S &= \sum\limits_{y=0}^N \exp\!\left(\frac{2\pi \mathrm{i} y c}{N+1} \right) \nonumber \\
    &= \frac{\exp\!\left(2\pi \mathrm{i} c \right) - 1}{\exp\!\left( \frac{2\pi \mathrm{i} c}{N+1} \right) - 1} \nonumber \\
    &=0.
\end{align}
For $c=0$, 
\begin{equation}
    S = N+1.
\end{equation}
Thus, putting the two cases together, we see that
\begin{equation}
\label{eq:summation-S-final}
    S = (N+1) \delta_{x, z}.
\end{equation}
Plugging $S$ back into \eqref{eq:final_state_triple_sum}, the final state is
\begin{align}
    \ket{\psi_f} &= \frac{1}{N+1} \sum\limits_{z=0}^N \sum\limits_{x=0}^N S \ket{z}_{C} \otimes P_x \ket{\phi} \nonumber \\
    &= \sum\limits_{z=0}^N \sum\limits_{x=0}^N  \delta_{x, z} \ket{z}_{C} \otimes P_x \ket{\phi} \nonumber \\
    &= \sum\limits_{x=0}^N \ket{x}_{C} \otimes P_x \ket{\phi}.
\end{align}

Thus, the probability of outcome $a \in \{0, \ldots, N\}$ when the $C$ register is measured is 
\begin{align}
    p(a) &= \left\Vert (\bra{a} \otimes I) \ket{\psi_f} \right\Vert_2^2 \nonumber \\
    &= \left\Vert P_a \ket{\phi} \right \Vert_2^2,
\end{align}
and the post-measurement state is given by
\begin{equation}
    \ket{\psi_a} = \frac{P_a \ket{\phi}}{\sqrt{p(a)}}. 
\end{equation}

Note that since $\ket{\phi} = \ket{\psi} \ket{0}^{\otimes w}$, and $\ket{\psi}$ can have a Hamming weight of at most $n$, the probability that $a > n$  is equal to zero. Thus, the algorithm realizes the coherent Hamming weight measurement on the input state~$\ket{\psi}$.
\end{proof}
\medskip

Step~3 of the algorithm requires the controlled-$U_y$ operation:
\begin{equation}
\label{eq:step3}
    \sum\limits_{y=0}^N \outerproj{y}_C \otimes U_y.
\end{equation}
Here, the group structure of $\{U_y\}_y$ is critical. We can replace this summation with a restricted sum over just the values of $y$ that are powers of two. The key insight is the fact that every $y$ in the sum can be written in terms of its binary representation. The group structure of the unitaries $\{U_y\}_y$ can then be leveraged to simplify the circuit. More concretely, we replace the summation with 
\begin{equation}
\label{eq:simpleStep3}
    \prod\limits_{m=1}^k (\outerproj{0}_{C_m} \otimes I + \outerproj{1}_{C_m} \otimes U_{2^{m-1}}),
\end{equation}
where $C_m$ denotes the $m$th qubit in register $C$. This is a product of $\mathcal{O}(\log(n))$ elements, each controlled on a single qubit. In contrast, the general sum given above is a sum over $\mathcal{O}(n)$ elements, with multiple control qubits. To show the equivalence between the two operations, we consider the action of \eqref{eq:simpleStep3} on a state of the form $\ket{a}\ket{\omega}$, where $a$ is a bitstring.
\begin{equation}
    \prod\limits_{m=1}^k (\outerproj{0}_{C_m} \otimes I + \outerproj{1}_{C_m} \otimes U_{2^{m-1}})\ket{a}\ket{\omega} .
\end{equation}
Now, expand $a$ in binary as $[a_1 a_2 \ldots a_k]$. We notice that if any $a_i = 1$, the modified Step~3 applies $U_{2^i}$. Thus, the overall effect is to apply the following unitary on the state $\ket{\omega}$:
\begin{align}
     \prod_{\substack{i=0 \\ a_i = 1}}^k U_{2^i} &= U_{\sum_{\substack{i=0 \\ a_i = 1}}^k 2^i} \nonumber \\
     &= U_{\sum_{i=0}^k a_i2^i} \nonumber \\
     &= U_a,
\end{align}
where the first equality holds because the unitaries all commute with each other, another property of the abelian group $C_{N+1}$. Thus, for each input bitstring $a$, the modified controlled-gates perform $U_a$ as needed.

\begin{remark}
Each unitary $U_y$ from \eqref{eq:Uy_new_def} consists of an overall phase and rotations on individual qubits. However, the algorithm uses these unitaries in a controlled manner, which makes the phases relative so that they cannot be ignored. To convert this relative phase into a global phase, we perform a rotation on the control qubit. The equivalence can be seen as follows:
\begin{align}
    &\outerproj{0}_C \otimes I + \outerproj{1}_C \otimes U_y \notag \\
    &= \outerproj{0}_C \otimes I + \outerproj{1}_C \otimes \left[ \exp\! \left(\frac{\mathrm{i} \pi y N }{N+1} \right) R_z\!\left(\frac{2\pi y}{N+1}\right)^{\otimes N}\right] \notag \\
    &= \left(R_z\!\left(\frac{\mathrm{i} \pi y N }{N+1} \right)\right)_C \times \notag \\
    & \qquad \left[ \outerproj{0}_C \otimes I + \outerproj{1}_C \otimes R_z\!\left(\frac{2\pi y}{N+1}\right)^{\otimes N}\right],
\end{align}
up to a global phase. In Figure~\ref{fig:Circuit_Rz}, we redraw the circuit to handle the relative phases appropriately.
\end{remark}

\begin{figure}
    \includegraphics[width=1.0\columnwidth]{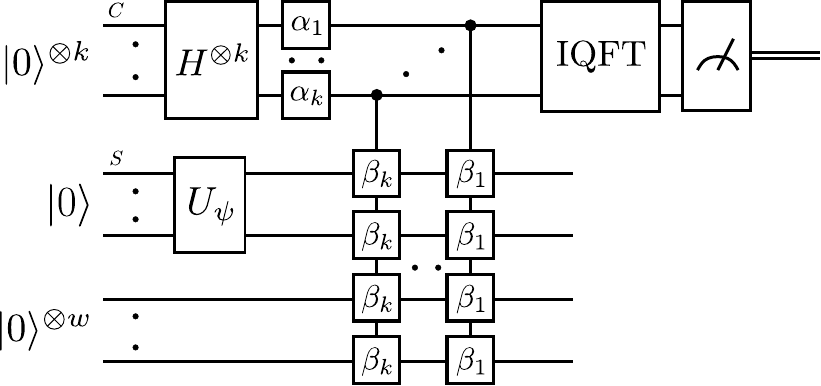}
    \caption{Practical implementation of the circuit from Figure~\ref{fig:Circuit_Uy}. The individual phase angles are pulled out as rotations acting on the control qubits. In the figure, $\alpha_j$ is a label for the rotation gate $R_z(\frac{\pi N}{N+1} 2^{j-1})$ and $\beta_j$ is a label for $R_z(\frac{2\pi}{N+1} 2^{j-1})$.}
    \label{fig:Circuit_Rz}
\end{figure} 

\subsection{Second Algorithm without Zero Padding}

We now present a more efficient algorithm that uses $\mathcal{O}(\log (n))$ control qubits, like Algorithm~\ref{alg:hamming-wt-meas-ineff}, but it instead has no need for extra qubits for zero padding. Note that in Figure~\ref{fig:Circuit_Rz}, there are multiple controlled-$R_z$ rotations acting on the zero-padded qubits. Using Lemma~\ref{lem:Rz-cancellation}, we can remove all these rotation gates and account for them in the form of a single $R_z$ gate on each control qubit. This leaves the zero-padded qubits in the zero state, and the entire algorithm does not affect these qubits. Thus, these qubits can be  eliminated. The circuit corresponding to the revised algorithm (Algorithm~\ref{alg:hamming-wt-meas-eff}) is depicted in Figure~\ref{fig:Circuit_Eff_Rz}.

\begin{Algorithm}
\label{alg:hamming-wt-meas-eff} 

\begin{figure}
    \includegraphics[width=1.0\columnwidth]{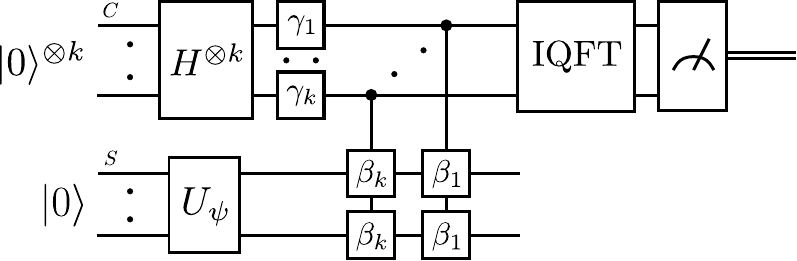}
    \caption{Efficient practical implementation of the circuit from Figure~\ref{fig:Circuit_Uy}. In the figure, $\gamma_j$ is a label for the rotation gate $R_z(\frac{n \pi}{N+1} 2^{j-1})$ and $\beta_j$ is a label for $R_z(\frac{2\pi}{N+1}2^{j-1})$.}
    \label{fig:Circuit_Eff_Rz}
\end{figure} 

The algorithm consists of the following steps:
\begin{enumerate}
\item Prepare a control register $C$, consisting of~$k$ qubits, in the state $|0\rangle_{C}$.

\item Act on register $C$ with $k$ Hadamard gates, leading to the state
\begin{align}
    \ket{+}_C &= H^{\otimes k}\ket{0}_C^{\otimes k} \nonumber \\
    &= \frac{1}{\sqrt{N+1}} \sum\limits_{y=0}^N \ket{y}_C.
\end{align}

\item Act on register $C$ with individual rotation gates $R_z(\gamma_j)$ on qubit $j$, labelled as $C_j$,
\begin{equation}
\label{eq:gamma_rotations}
    \bigotimes_{j=1}^k R_z(\gamma_j)_{C_j}
\end{equation}
where the elements of the set $\{\gamma_j\}_{j=1}^k$ of angles  are defined as
\begin{equation}
    \gamma_j \coloneqq \frac{n \pi}{N+1} 2^{j-1}.
\end{equation}

\item Append the state of interest $\ket{\psi}_S$ and perform the following controlled unitary:
\begin{equation}
\label{eq:controlled-Rz-eff}
\prod_{j=1}^k \left[ \outerproj{0}_{C_j} \otimes I_S + \outerproj{1}_{C_j} \otimes R_z^{\otimes n}(\beta_j)_S \right]   ,
\end{equation}
where the elements of the set $\{\beta_j\}_{j=1}^k$ of angles   are defined as follows:
\begin{equation}
    \beta_j \coloneqq \frac{2\pi}{N+1}2^{j-1}.
\end{equation}

\item Perform an inverse quantum Fourier transform on register~$C$ and measure in the computational basis $\{\outerproj{0}, \ldots, \outerproj{N} \}$.
\end{enumerate}
\end{Algorithm}

\begin{theorem}
Algorithm~\ref{alg:hamming-wt-meas-eff} realizes the coherent Hamming-weight measurement $\{P_x\}_{x=0}^{n}$ on the state $\ket{\psi}$.
\end{theorem}

\begin{proof}
After Step~2 of the algorithm, the state is given by $\ket{+}_C$. To see the action of the gates \eqref{eq:gamma_rotations} on this state, we first find the action on a single computational basis state $\ket{y}_C \equiv \ket{y_k}_{C_k} \ldots \ket{y_1}_{C_1}$.
\begin{align}
    \ket{y}_C &\rightarrow \bigotimes_{j=1}^k R_z(\gamma_j)_{C_j} \ket{y_k}_{C_k} \ldots \ket{y_1}_{C_1} \notag \\
    &= \exp{\!\left( \sum_{j=1}^k \frac{\mathrm{i} \gamma_j}{2}(2y_j - 1) \right)} \ket{y_k}_{C_k} \ldots \ket{y_1}_{C_1} \notag \\
    &= \exp\!\left( \sum_{j=1}^k \mathrm{i} \gamma_j y_j - \frac{\mathrm{i}}{2} \sum_{j=1}^k\gamma_j \right)
    \ket{y_k}_{C_k} \ldots \ket{y_1}_{C_1} \notag \\
    &= \exp\!\left( \mathrm{i} \dopr{\gamma}{y}  - \frac{\mathrm{i}}{2} \Gamma \right) \ket{y}_C,
\end{align}
where the first equality is due to Remark~\ref{rmk:Rz_computational} and $\Gamma$ is defined as
\begin{equation}
    \Gamma \coloneqq \sum_{j=1}^k \gamma_j.
\end{equation}
Thus, the result of Step~3 is given by
\begin{equation}
    \frac{1}{\sqrt{N+1}} \sum_{y=0}^N \exp{\!\left( \mathrm{i} \dopr{\gamma}{y}  - \frac{\mathrm{i}}{2} \Gamma \right)} \ket{y}_C.
\end{equation}

Next, we see the action of the gates \eqref{eq:controlled-Rz-eff} on this state. Similar to the previous discussion, we first find the action when controlled on a single computational basis state $\ket{y}_C \equiv \ket{y_k}_{C_k} \cdots \ket{y_1}_{C_1}$.
\begin{align}
    \ket{y}_C \ket{\psi}_S 
    &\rightarrow \ket{y}_C R_z^{\otimes n}(y_1 \beta_1) \cdots R_z^{\otimes n}(y_k \beta_k) \ket{\psi}_S \notag \\
    &= \ket{y}_C R_z^{\otimes n} (y_1\beta_1 + \cdots + y_k \beta_k) \ket{\psi}_S \notag \\
    &= \ket{y}_C R_z^{\otimes n} \!\left(\sum_{j=1}^k y_j\beta_j \right) \ket{\psi}_S \notag \\
    &= \ket{y}_C R_z^{\otimes n} \!\left(\dopr{y}{\beta} \right) \ket{\psi}_S,
\end{align}
where the first step is due to the reasoning from Remark~\ref{rmk:C-Rz_computational}.

Thus, the result of Step~4 is given by
\begin{equation}
    \frac{1}{\sqrt{N+1}} \sum_{y=0}^N \exp{\!\left( \mathrm{i} \dopr{\gamma}{y}  - \frac{\mathrm{i}}{2} \Gamma \right)} \ket{y}_C \otimes R_z^{\otimes n} \left(\dopr{y}{\beta} \right) \ket{\psi}_S
\end{equation}

To simplify, we now expand the two dot products in the previous expression.
\begin{align}
    \dopr{\gamma}{y} &= \sum_{j=1}^k \gamma_j y_j =\frac{n \pi}{N+1} \sum_{j=1}^k y_j 2^{j-1} = \frac{n \pi y}{N+1}, \\
    \dopr{\beta}{y} &= \sum_{j=1}^k \beta_j y_j =\frac{2 \pi}{N+1} \sum_{j=1}^k y_j 2^{j-1} = \frac{2 \pi y}{N+1}.
\end{align}
Thus, the state can be rewritten as
\begin{multline}
    \frac{1}{\sqrt{N+1}} \exp{\!\left(  - \frac{\mathrm{i}}{2} \Gamma \right)} \sum_{y=0}^N \exp{\!\left( \frac{\mathrm{i}n\pi y}{N+1} \right)} \ket{y}_C \\ \otimes R_z^{\otimes n} \!\left( \frac{2 \pi y}{N+1}\right) \ket{\psi}_S.
    \label{eq:proof-step-key}
\end{multline}
Next, using \eqref{eq:Rz-phi-otimes-N}, we see that
\begin{equation}
    R_z^{\otimes n}\!\left(\frac{2 \pi y}{N+1}\right) = \exp\! \left(\frac{-\mathrm{i}n \pi y}{N+1} \right) \sum_{x=0}^n \exp\! \left( \frac{2\pi \mathrm{i} xy}{N+1} \right) P_x.
\end{equation}
Plugging this equation into \eqref{eq:proof-step-key} and simplifying, we see that the state after Step~4 can thus be expressed as
\begin{multline}
    \frac{1}{\sqrt{N+1}} \exp{\!\left(  - \frac{\mathrm{i}}{2} \Gamma \right)} \sum_{y=0}^N \sum_{x=0}^n \exp\! \left( \frac{2\pi \mathrm{i} xy}{N+1} \right) \\ \ket{y}_C P_x \ket{\psi}_S.
\end{multline}
Now applying the $\operatorname{IQFT}$ to the control register, the state finally becomes
\begin{multline}
 \frac{1}{N+1} \exp{\!\left(  - \frac{\mathrm{i}}{2} \Gamma \right)} \sum_{z,y,x=0}^N  \exp\! \left( \frac{2\pi \mathrm{i} (x-z) y}{N+1} \right) \\ \ket{z}_C P_x \ket{\psi}_S.
\end{multline}
Following the reasoning from \eqref{eq:summation-S-init}--\eqref{eq:summation-S-final}, the final form is then
\begin{equation}
     \exp{\!\left(  - \frac{\mathrm{i}}{2} \Gamma \right)} \sum_{x=0}^N \ket{x}_C \otimes P_x\ket{\psi}_S,
\end{equation}
where $\exp{\!\left(  - \frac{\mathrm{i}}{2} \Gamma \right)}$ is a global phase.

Similar to Algorithm~\ref{alg:hamming-wt-meas-ineff}, the probability of outcome $a \in \{0, \ldots, N\}$ when the $C$ register is measured is 
\begin{equation}
    p(a) = \left\Vert P_a \ket{\psi} \right\Vert_2^2,
\end{equation}
and the post-measurement state is given by
\begin{equation}
    \ket{\psi_a} = \frac{P_a \ket{\psi}}{\sqrt{p(a)}}. 
\end{equation}
Thus, the algorithm realizes a coherent Hamming-weight measurement of the state $\ket{\psi}$.
\end{proof}

\section{Depth-Width Tradeoff and 2D Architecture Implementation}

\label{sec:dep-wid-tradeoff}

In this section, we provide a depth-width tradeoff for Algorithm~\ref{alg:hamming-wt-meas-eff}. We show that the construction above is the width-optimal version, i.e., using the minimal number of control qubits. The tradeoff construction provides a way to use more control qubits with the benefit of a significant reduction of the depth of the overall circuit. This approach was used recently for the task of multivariate trace estimation~\cite{QKW24}, and we show how it can be used here also. It in turn makes use of ideas from fault-tolerant quantum computing, known as Shor error correction~\cite{S96} (see also Figure~2 of~\cite{G10}). 

\begin{figure*}
    \includegraphics[width=0.9\textwidth]{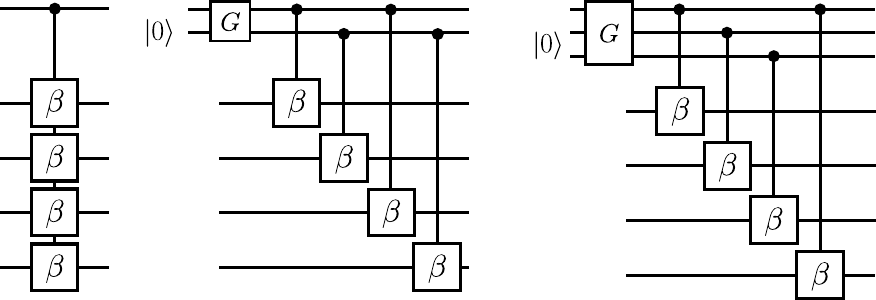}
    \caption{Showing the depth-width tradeoff for an $n=4$ qubit state. The three figures have $s=1, 2, 3$ respectively. The circuit $G$, involving mid-circuit measurements and feedback, encodes the input state into a repetition code and has a constant quantum depth~\cite{QKW24}.}
    \label{fig:GHZ_replacement}
\end{figure*}

We first note that there are~$k$ sets of controlled-$R_z$ gates of the form
\begin{equation}
    \left[ \outerproj{0}_{C_j} \otimes I_S + \outerproj{1}_{C_j} \otimes R_z^{\otimes n}(\beta_j)_S \right]   
\end{equation}
for $j \in \{1, \ldots, k\}$. Consider the case $j=1$, i.e., the set of controlled rotations that are controlled on qubit $C_1$,
\begin{equation}
    \left[ \outerproj{0}_{C_1} \otimes I_S + \outerproj{1}_{C_1} \otimes R_z(\beta_j)_{S_i} \right],
\end{equation}
for $i \in \{1, \ldots, n\}$. There are $n$ controlled rotations that cannot be parallelized because they use the same control qubit. Thus, the depth of this construction is linear, and combining the rotations using all the control qubits, we see that the overall construction has depth $\mathcal{O}(n \log (n))$. 

However, the tensor-product structure of $R_z^{\otimes n}(\beta_j)$ allows for using extra control qubits to reduce the depth. Once again, consider the first control qubit $C_1$. The control qubit $C_1$ can be encoded into a repetition code on $s$ qubits. More concretely, we tensor in $(s-1)$ qubits $A_2, \ldots, A_s$ in the $\ket{0}$ state, and the control qubit $C_1$ is expanded to a control register $\bar{C}_1 \equiv C_1 A_2 \ldots A_s$, as follows:
\begin{multline}
    \left[\alpha \ket{0}_{C_1} + \beta\ket{1}_{C_1}\right] \ket{0}_{A_2} \cdots \ket{0}_{A_s} \\
    \rightarrow \alpha \ket{00 \ldots 0}_{\overline{C}_1} + \beta\ket{11\ldots 1}_{\overline{C}_1}.
\end{multline}
This encoding can be achieved in constant quantum depth by allowing for mid-circuit measurements and feedback~\cite{QKW24}.

Next, the set of $n$ controlled rotations, which were all controlled on the same qubit $C_1$, are now split up among the $s$ different control qubits $C_1, A_2, \ldots, A_s$. Since the controlled rotations now act on different qubits, they can be effectively parallelized. We show the construction for $n=4$ and varying $s$ values in Figure~\ref{fig:GHZ_replacement}.

To round out the algorithm, the $\operatorname{IQFT}$ and measurement need to be translated as well. Alternatively, we need to specify how to perform a repetition-code encoded $\operatorname{IQFT}$, measurement, and classical postprocessing to get the same measurement outcomes as the un-encoded algorithm. We now define the semi-classical encoded IQFT, making use of original insights from~\cite{PhysRevLett.76.3228} (see also~\cite[Figure~1]{Smolin2013} for a visual depiction). 

\begin{figure*}[htbp]
    \includegraphics[width=0.9\textwidth]{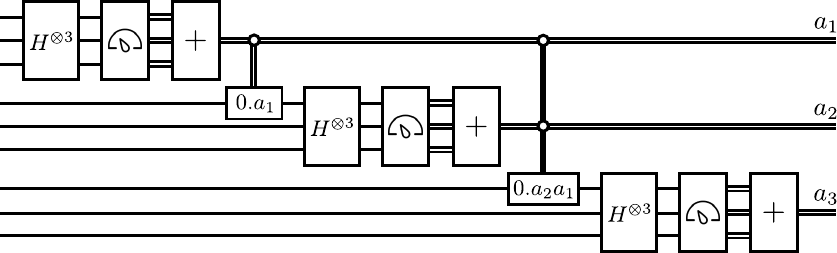}
    \caption{Semi-classical encoded IQFT $F^{\text{sc}}_{k, s}$ for~$k=3$ and $s=3$. The operations labelled by a binary expansion denotes a $Z$-rotation. For example $0.a_2 a_1$ denotes the following gate $R_z(\pi\ 0.a_2 a_1)$.}
    \label{fig:circuit_sc_enc_IQFT}
\end{figure*}

\begin{figure*}[htbp]
    \includegraphics[width=0.95\textwidth]{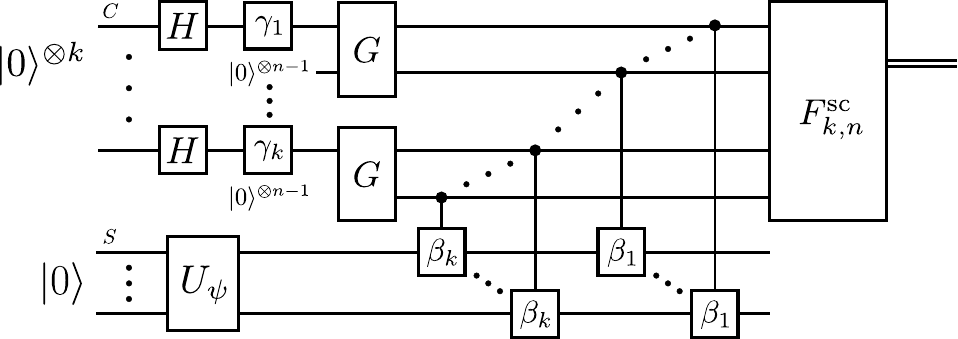}
    \caption{The depth-optimal circuit construction to perform the coherent Hamming weight measurement on the $n$-qubit input state $\ket{\psi} = U_{\psi} \ket{0}$. As discussed in Section~\ref{sec:dep-wid-tradeoff}, the transformation $G$ encodes the input state into a repetition code and has a constant quantum depth~\cite{QKW24}.}
    \label{fig:full-circuit}
\end{figure*}

\begin{figure*}
\includegraphics[width=0.18\textwidth,angle=270]{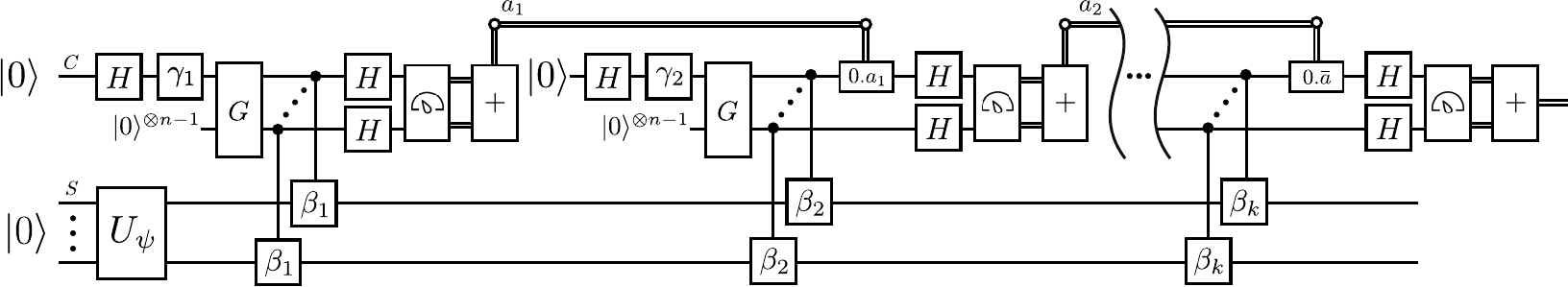}
    \caption{Variation of Figure~\ref{fig:full-circuit} when qubit resets are available. The overall circuit now  needs only $n$ control qubits and has depth $\mathcal{O}(\log n)$.}
    \label{fig:full-circuit-resets}
\end{figure*}

\begin{figure*}[htbp]
    \includegraphics[width=\textwidth]{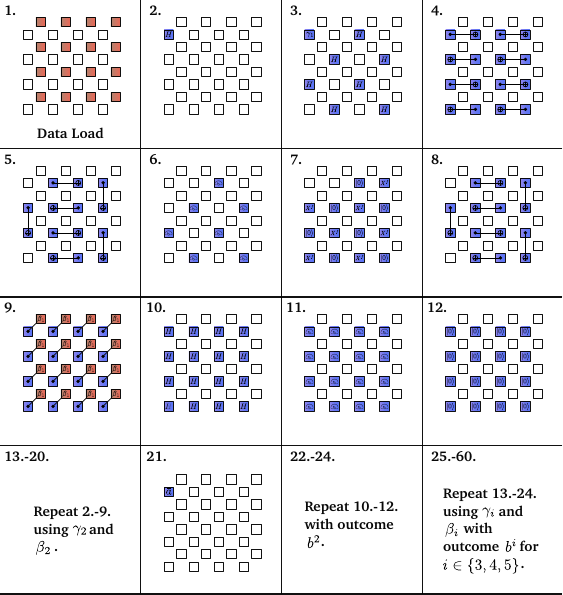}
    \caption{$2$-D implementation of the circuit given in Figure~\ref{fig:full-circuit-resets}. We show the diagram for a $16$-qubit input state $\ket{\psi}$. Furthermore, we set $s=16$ to get the depth-optimal version of the algorithm. In the different operations, the label $\bar{a}$ refers classical controlled $Z$-rotation $R_z\!\left( \frac{\pi}{2^{k-1}}  \bar{a} \right)$ where $\bar{a} \coloneqq a_{k-1} \cdots a_1$. Each bit $a_i$ is the sum of entries of the bitstring $b^i$ modulo $2$.}
    \label{fig:full-circuit-2d}
\end{figure*}

\begin{definition}[Semi-classical encoded IQFT]
    The semi-classical encoded IQFT $F^{\text{sc}}_{k, s}$ acts on a~$k$-qubit quantum state encoded into an $s$-qubit repetition code and has a recursive circuit description consisting of the following steps:
    \begin{enumerate}
        \item Apply $F^{\text{sc}}_{k-1, s}$ on the first $(k-1)s$ qubits, and obtain the measurement outcome $a_1 \cdots a_{k-1}$.
        
        \item Based on the classical data, define the decimal value $\bar{a} \coloneqq a_{k-1} \cdots a_1$, and apply the controlled rotation $R_z\!\left( \frac{\pi}{2^{k-1}}  \bar{a} \right)$ on qubit  $[(k-1)s + 1]$.
        \item For $j \in \{1, \ldots, s\}$, apply the Hadamard gate $H$ on qubit $(k-1)s + j$.
        
        \item For $j \in \{1, \ldots, s\}$, measure qubit $(k-1)s + j$ to get the classical bits $b^k_{1} \cdots b^k_{s}$.
        
        \item Add the classical bits modulo $2$ to obtain $a_k \coloneqq b^k_{1} + \cdots + b^k_{s}$ .
        
        \item Output $a_1, a_2, \ldots, a_k$.
    \end{enumerate}

\end{definition}

The construction for $F^{\text{sc}}_{3, 3}$ is shown in Figure~\ref{fig:circuit_sc_enc_IQFT}.

\begin{theorem}
    The semi-classical encoded IQFT $F^{\text{sc}}_{k, s}$ acting on an encoded state gives the same measurement outcome probabilities as performing the IQFT on the original quantum state and measuring in the computational basis. 
\end{theorem}

\begin{proof}
    Consider an $s$-qubit repetition code that encodes the basis states as follows:
    \begin{align}
        \ket{0} &\rightarrow \ket{\bar{0}} \equiv  \ket{0}^{\otimes s}, \\
        \ket{1} &\rightarrow \ket{\bar{1}} \equiv \ket{1}^{\otimes s}.
    \end{align}
    The semi-classical IQFT consists of Hadamard gates and classically controlled $Z$-rotations. First consider that $Z$-rotations acting on just one of the encoded qubits realizes the same operation as the rotation on the un-encoded qubit:
    \begin{align}
        \bar{R}_z(\phi) (\alpha \ket{\bar{0}} + \beta \ket{\bar{1}}) &= \alpha \ket{\bar{0}} + e^{ \mathrm{i} \phi} \beta \ket{\bar{1}}, \\
        R_z(\phi) (\alpha \ket{0} + \beta \ket{1}) &= \alpha \ket{0} + e^{ \mathrm{i} \phi} \beta \ket{1}.
    \end{align}

    We now replace the Hadamard gates and measurement with a tensor-product Hadamard operation, measurement, and classical post processing. Consider the action of a Hadamard gate and a computational basis measurement on an arbitrary state. The probability of outcome $a \in \{0,1\}$ is given by
    \begin{align}
        p(a) &= \vert \langle a \vert H (\alpha \ket{0} + \beta \ket{1}) \vert^2  \\
        &= \frac{1}{2} \left[ 1 + (-1)^a \cdot  2\operatorname{Re}(\alpha \beta^*)\right].
    \end{align}

    Now, as seen in Step~3, in the semi-classical encoded IQFT, we apply a Hadamard gate on each qubit and then measure in the computational basis. The probability of obtaining outcomes $b_1 \cdots b_s$ is given by
    \begin{align}
        p(b_1 \cdots b_s) &= \vert \langle b_1 \cdots b_s \vert H^{\otimes s} (\alpha \ket{\bar{0}} + \beta \ket{\bar{1}}) \vert^2 \notag \\
        &= \vert \langle b_1 \cdots b_s \vert H^{\otimes s} (\alpha \ket{0}^{\otimes s} + \beta \ket{1}^{\otimes s}) \vert^2 \notag \\
        &= \frac{1}{2^s} \left[ 1 + (-1)^{b_1 + \cdots + b_s} \cdot  2\operatorname{Re}(\alpha \beta^*)\right].
    \end{align}

    Lastly, we define $a = b_1 + \cdots + b_s$ modulo 2. Then, 
    \begin{align}
        p(a) &= \sum\limits_{b_1, \ldots, b_s} p(a \vert b_1, \ldots, b_s) p(b_1, \ldots b_s) \notag \\
        &= \sum\limits_{b_1, \ldots, b_s} \delta(a, b_1 + \cdots + b_s) p(b_1, \ldots, b_s) \notag \\
        &= \frac{1}{2} \left[ 1 + (-1)^a \cdot  2\operatorname{Re}(\alpha \beta^*)\right].
    \end{align}

    Thus, measuring all the encoded qubits and setting $a_k = b^k_1 + \ldots + b^k_s$ leads to the same probabilities as the unencoded operation. 
\end{proof}

\begin{remark}
    The depth of the semi-classical encoded IQFT $F^{\text{sc}}_{k, s}$ is linear in~$k$ and independent of $s$.
\end{remark}
\begin{proof}
    On each logical qubit, $F^{\text{sc}}_{k, s}$ consists of two operations: $H^{\otimes s}$ and a classically controlled rotation. In the traditional IQFT, the rotations are quantum-controlled and act on all qubits before, leading to an $\mathcal{O}(k^2)$ depth. In the semi-classical version, we get a compiled rotation gate, leading to a linear depth.
\end{proof}

\medskip
The depth-optimal version of our algorithm for an $n$-qubit input state is shown in Figure~\ref{fig:full-circuit}. An important point to note here is that, even though Figure~\ref{fig:full-circuit} depicts $nk$ control qubits, if qubit resets are available, we only need $n$ control qubits. This is because each step can be performed in a sequential manner as follows and in constant quantum depth (see Figure~\ref{fig:full-circuit-resets} for this variation): the first $n$ qubits are prepared in a GHZ state, then the controlled rotations labeled by $\beta_1$ are performed, then the first step of the semiclassical encoded IQFT is performed, leading to a measurement outcome $a_1$. Then the $n$ control qubits are reset, prepared in a GHZ state, the controlled rotations labeled by $\beta_2$ are performed, then the second step of the semiclassical encoded IQFT is performed (with an action conditioned on $a_1$). There are $k-2$ more steps like this, thus leading to a $\mathcal{O}(\log n)$ depth circuit with $n$ control qubits.

An implementation of the depth-optimal version on a 2D architecture is given in Figure~\ref{fig:full-circuit-2d}. We note that the overall growth of the depth is still $\mathcal{O}(\log n)$ when restricted to a 2D ladder architecture. One leg of the ladder consists of ancilla qubits, and the other leg consists of the data qubits. 

\section{Resource Counting}

In Table~\ref{tab:resource_count}, we delineate various resource counts for the circuits for different $s$ values. We note that for the different estimates, we assume that qubit resets are available. The depth of the semiclassical encoded $\operatorname{IQFT}$ is $\mathcal{O}(k)$, where~$k = \mathcal{O}(\log n)$ and $n$ is the number of qubits of the input state $\ket{\psi}$. 

\begin{table*}
    \centering
    \begin{tabular}{c|c|c|c}
        & Qubit-Optimal & Arbitary & Depth-optimal \\
        & $s=1$ & $1 < s < n$ & $s=n$ \\
        \hline \hline
        Total Depth & $\mathcal{O}(n \log n)$ & $\mathcal{O}(\lceil \frac{n}{s} \rceil \log n )$ & $\mathcal{O}(\log n)$\\
        \# of Control Qubits & $1$ & $ s$ & $n$\\
    \end{tabular}
    \caption{Resource count for the algorithm construction for varying $s$ values, the latter of which characterizes the depth-width tradeoff. We note that for the different estimates, we assume that qubit resets are available.}
    \label{tab:resource_count}
\end{table*}

\section{Simulation Results}

In this section, we report the results of simulations of our algorithm for a wide range of examples. We benchmarked our algorithm against the true distribution of Hamming weights. For each example, we considered a noiseless simulation, a simulation with shot noise, and a noisy quantum simulation. Lastly, we also ran the algorithm on actual quantum devices. The results are depicted in Figures~\ref{fig:1st-example}--\ref{fig:last-example}.

In most of the cases, we simulated the width-optimal version of our construction as the depth-optimal version requires a larger number of control qubits and mid-circuit measurements, which are a constraint on most available hardware today. However, for a three-qubit example, we also classically simulated it using different $s$ values and the semi-classical encoded IQFT.

For the shot-noise simulations, each run used $10,000$ shots. For the noisy simulations, we used the 16-qubit IBM simulator {\tt FakeGuadalupeV2}. Each run consisted of $10,000$ shots and the median $T_1$ and $T_2$ values were $7.2 \times 10^{-5}$ and $8.8 \times 10^{-5}$, respectively. For the quantum device runs, we used the 127-qubit IBM device {\tt ibm$\_$brisbane}. Each run consisted of $10,000$ shots and the median $T_1$ and $T_2$ values are $2.2 \times 10^{-4}$ and $1.3 \times 10^{-4}$, respectively.

All the program code is available as arXiv ancillary files along with the arXiv posting of this paper.

\begin{figure*}[htbp]{}
    \centering 
    \begin{subfigure}[b]{0.49\textwidth}   
    \centering
    \includegraphics[width=1\textwidth]{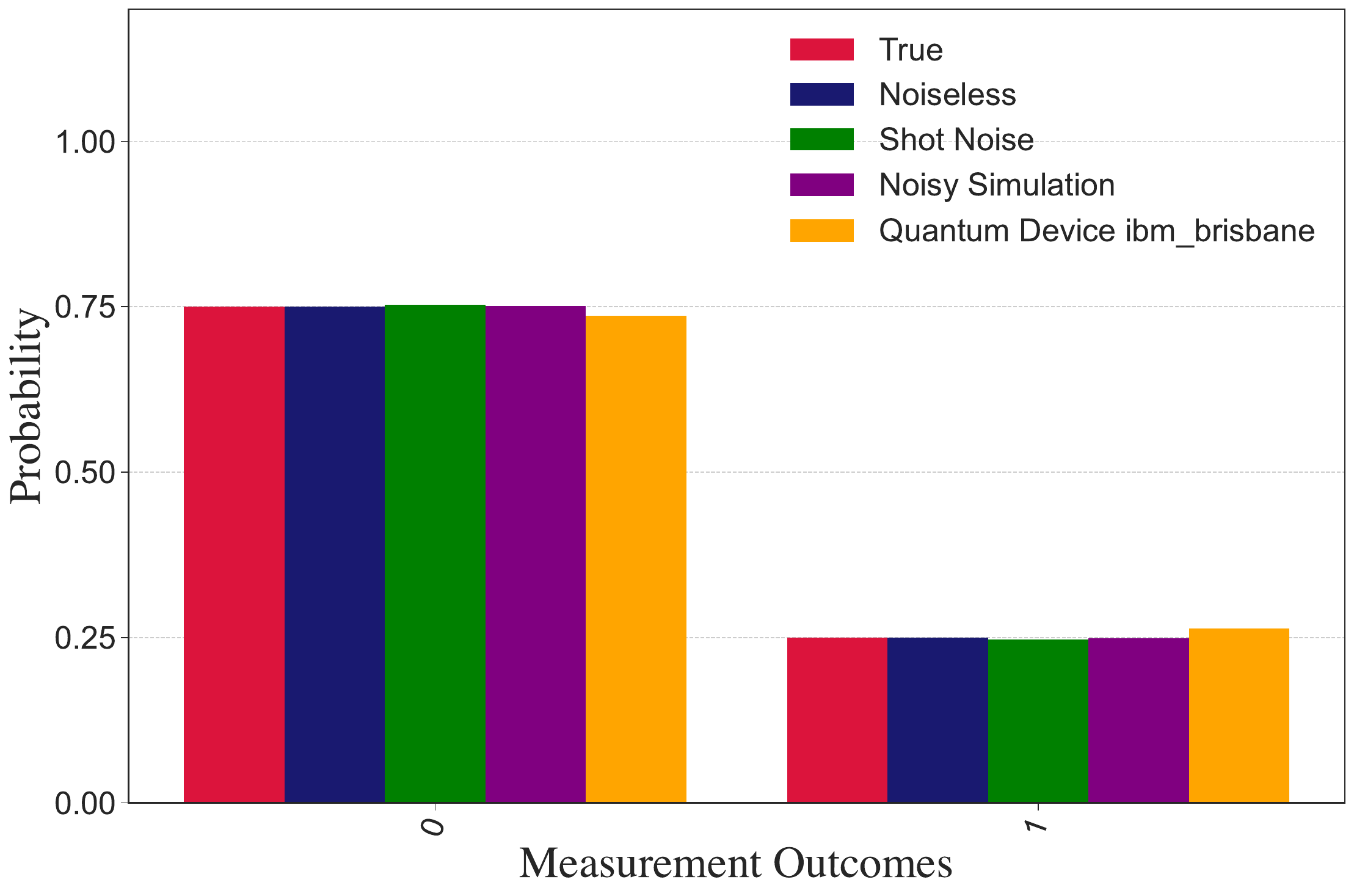}
    \end{subfigure}
    \hfill
    \begin{subfigure}[b]{0.49\textwidth}
    \centering
    \includegraphics[width=1\textwidth]{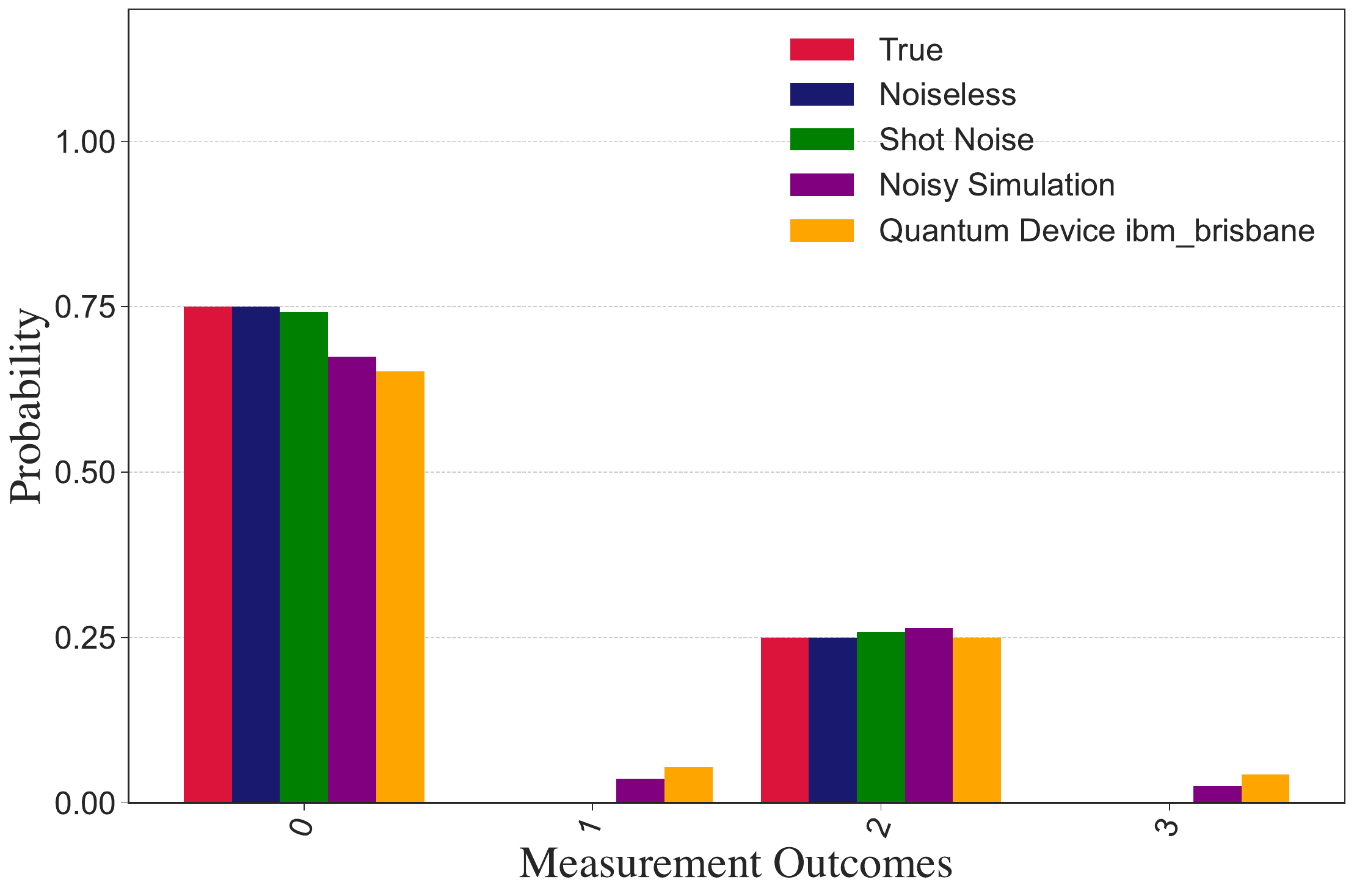}
    \end{subfigure} 
    \hfill
    \vspace{0.5cm}
    \begin{subfigure}[b]{0.49\textwidth}
    \includegraphics[width=1\textwidth]{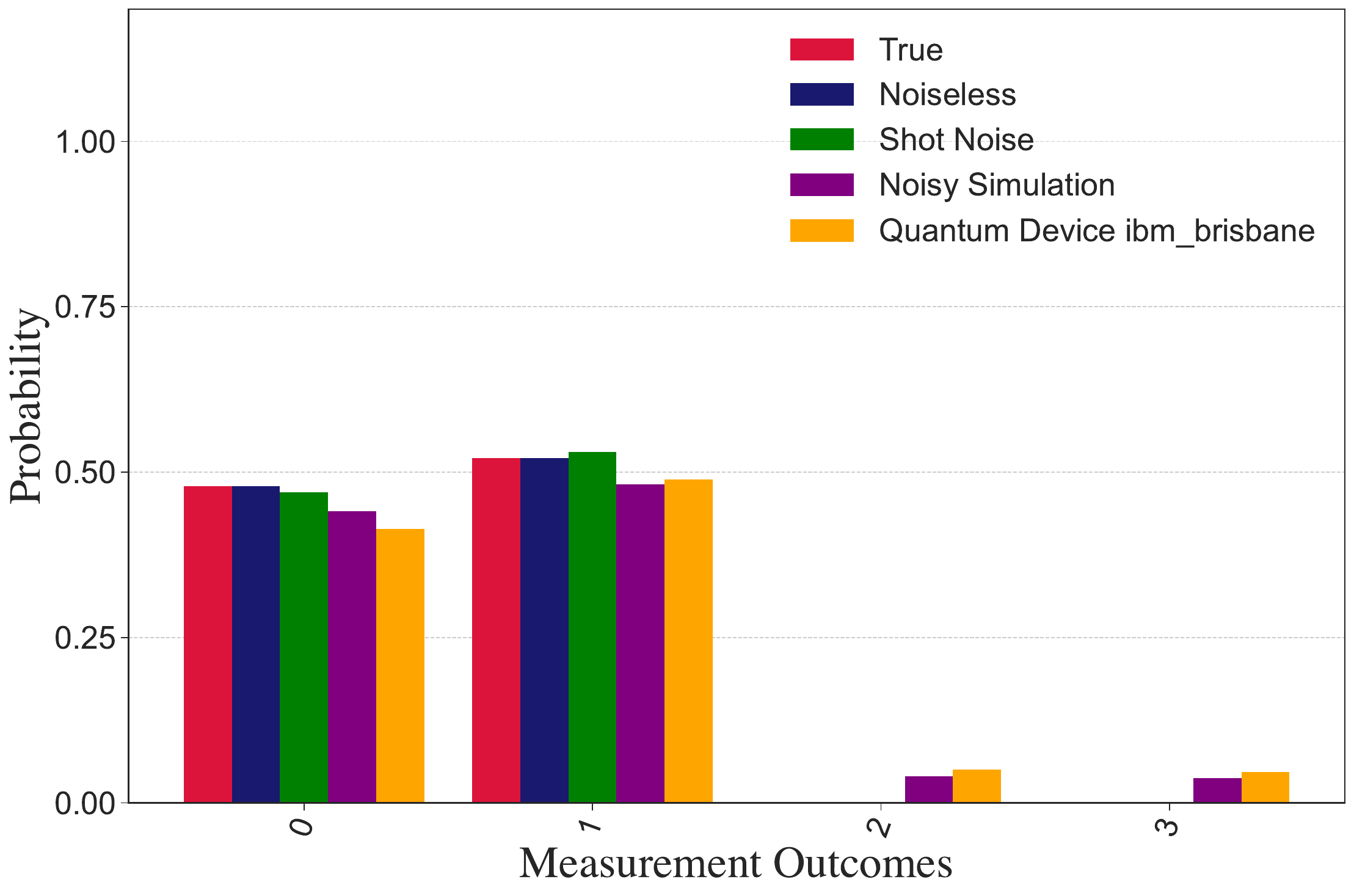}
    \end{subfigure} 
\captionsetup{justification=centering}
\caption{Distribution of counts obtained for the input state $\ket{\psi} = R_y\!\left(\frac{\pi}{3}\right) \ket{0}$ (top left),  the input state $\ket{\psi} = \frac{1}{2}\left( \sqrt{3}\ket{00} + \ket{11}\right)$ (top right), and a randomly generated two-qubit input state (bottom).}
\label{fig:1st-example}
\end{figure*}

\begin{figure*}[htbp]
    \centering
    \begin{subfigure}[b]{0.49\textwidth}
    \centering
    \includegraphics[width=1.0\textwidth]{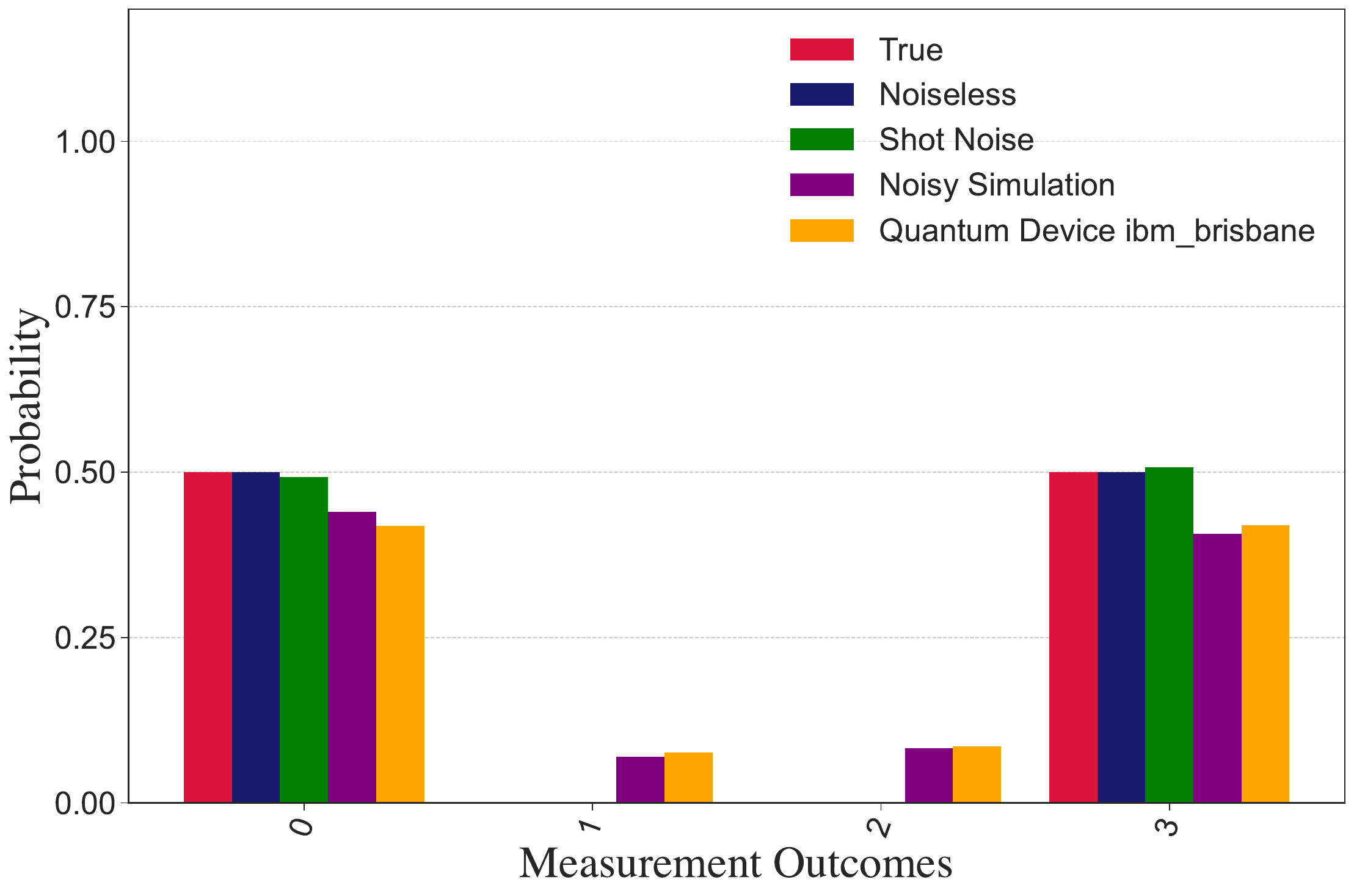}
    \end{subfigure}
    \hfill
    \begin{subfigure}[b]{0.49\textwidth}
    \centering
    \includegraphics[width=1.0\textwidth]{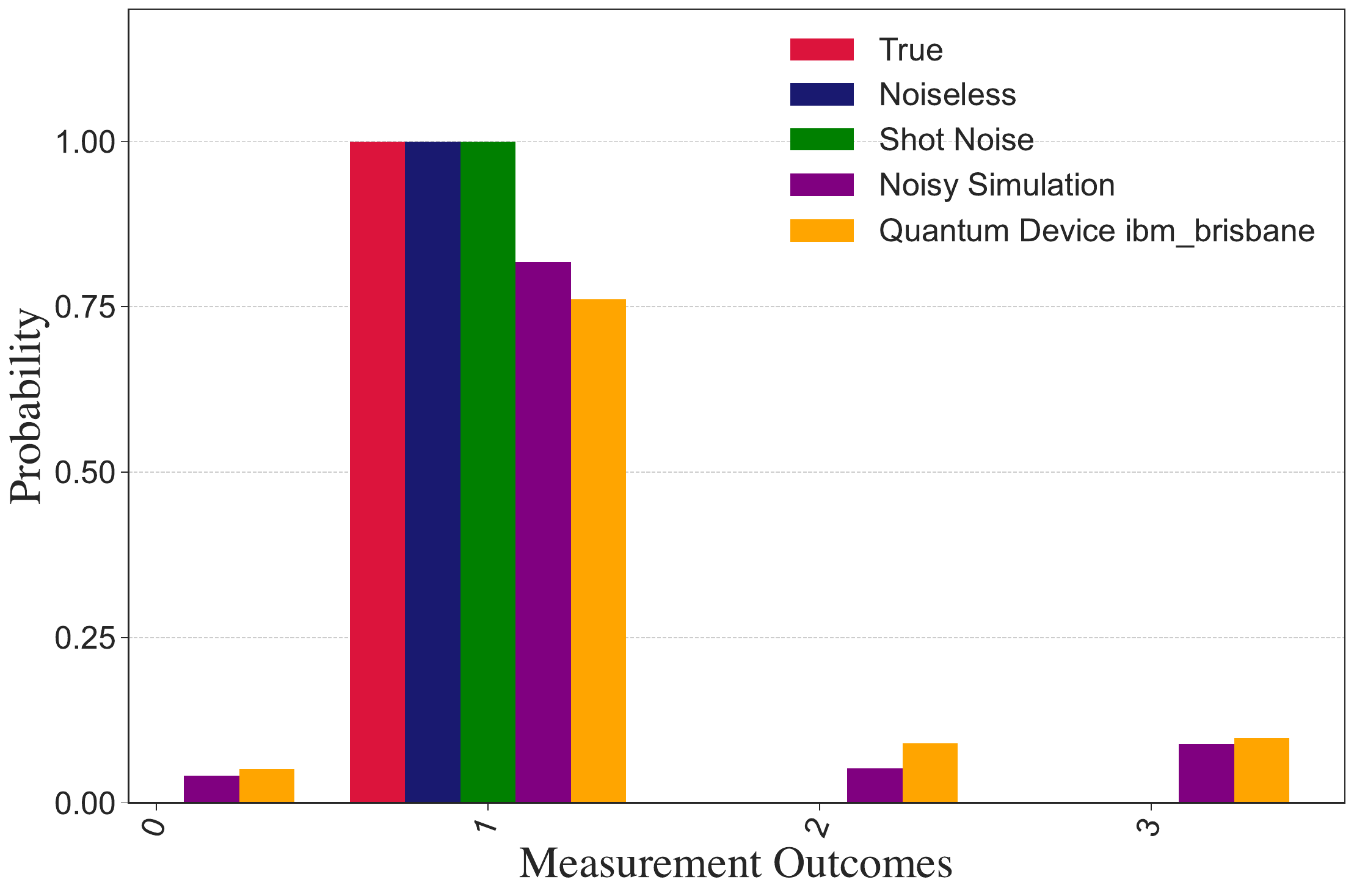}
    \end{subfigure} 
    \hfill
    \vspace{0.5cm}
    \begin{subfigure}[b]{0.49\textwidth}
    \includegraphics[width=1.0\textwidth]{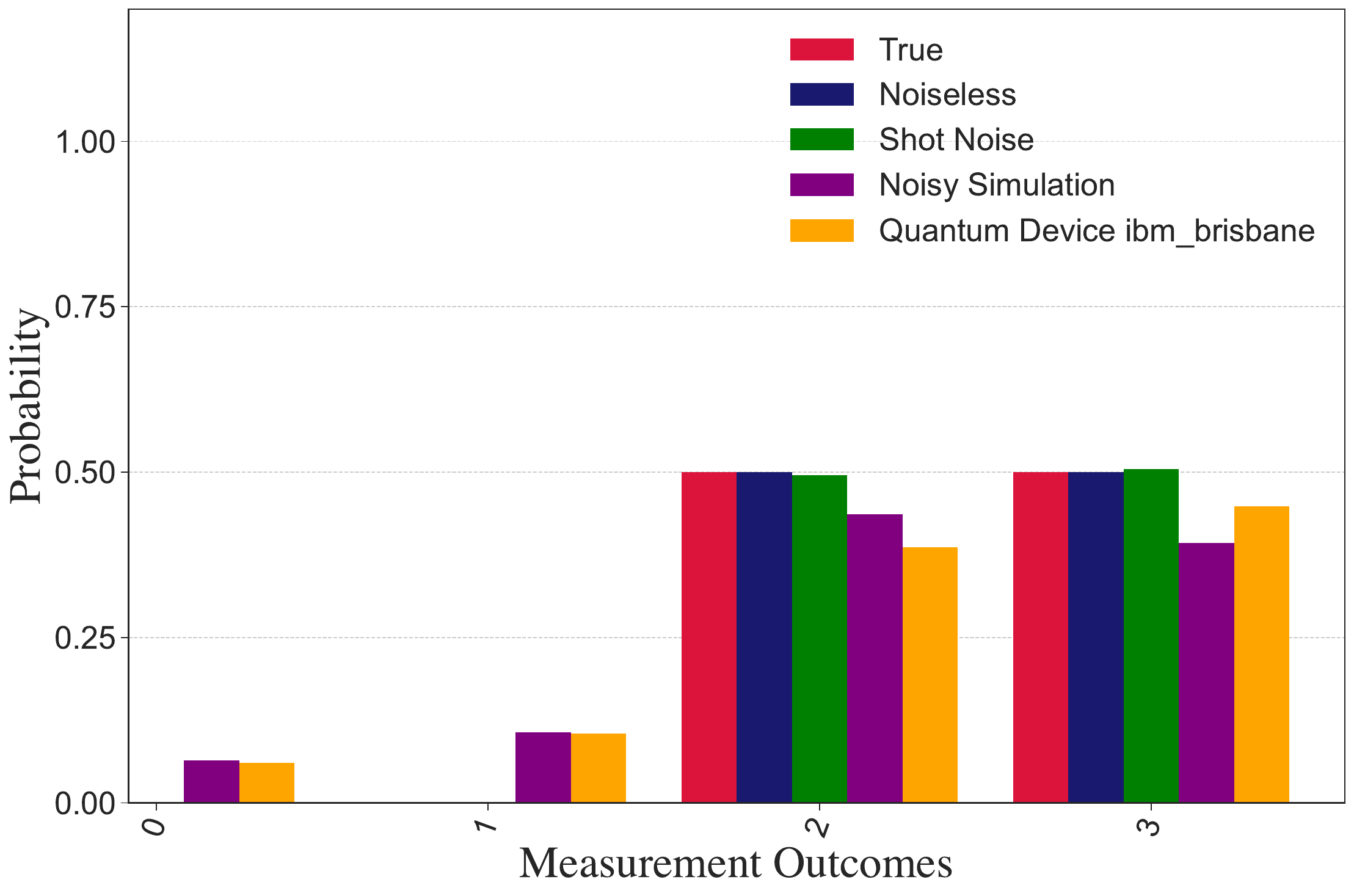}
    \end{subfigure} 
    \captionsetup{justification=centering}
    \caption{Distribution of counts obtained for the input state $\ket{\operatorname{GHZ}} = \frac{1}{\sqrt{2}}\left( \ket{000} + \ket{111}\right)$ (top left), the input state $\ket{\operatorname{W}} = \frac{1}{\sqrt{3}}\left( \ket{001} + \ket{010} + \ket{100}\right)$ (top right), and a randomly generated three-qubit input state (bottom).}
\end{figure*}

\begin{figure*}[htbp]
    \centering
    \begin{subfigure}[b]{0.49\textwidth}
    \centering
    \includegraphics[width=1.0\textwidth]{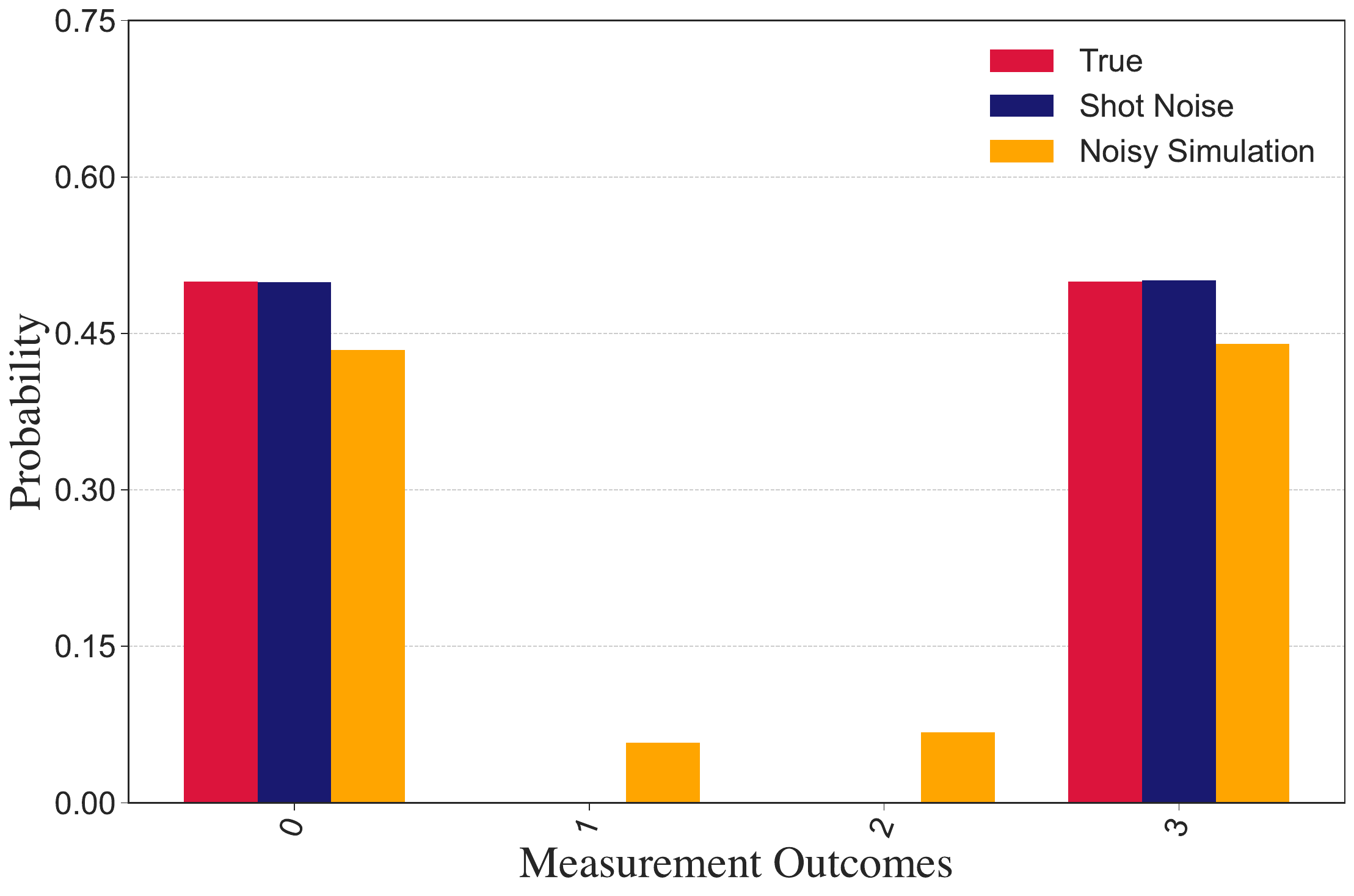}
    \end{subfigure}
    \hfill
    \begin{subfigure}[b]{0.49\textwidth}
    \centering
    \includegraphics[width=1.0\textwidth]{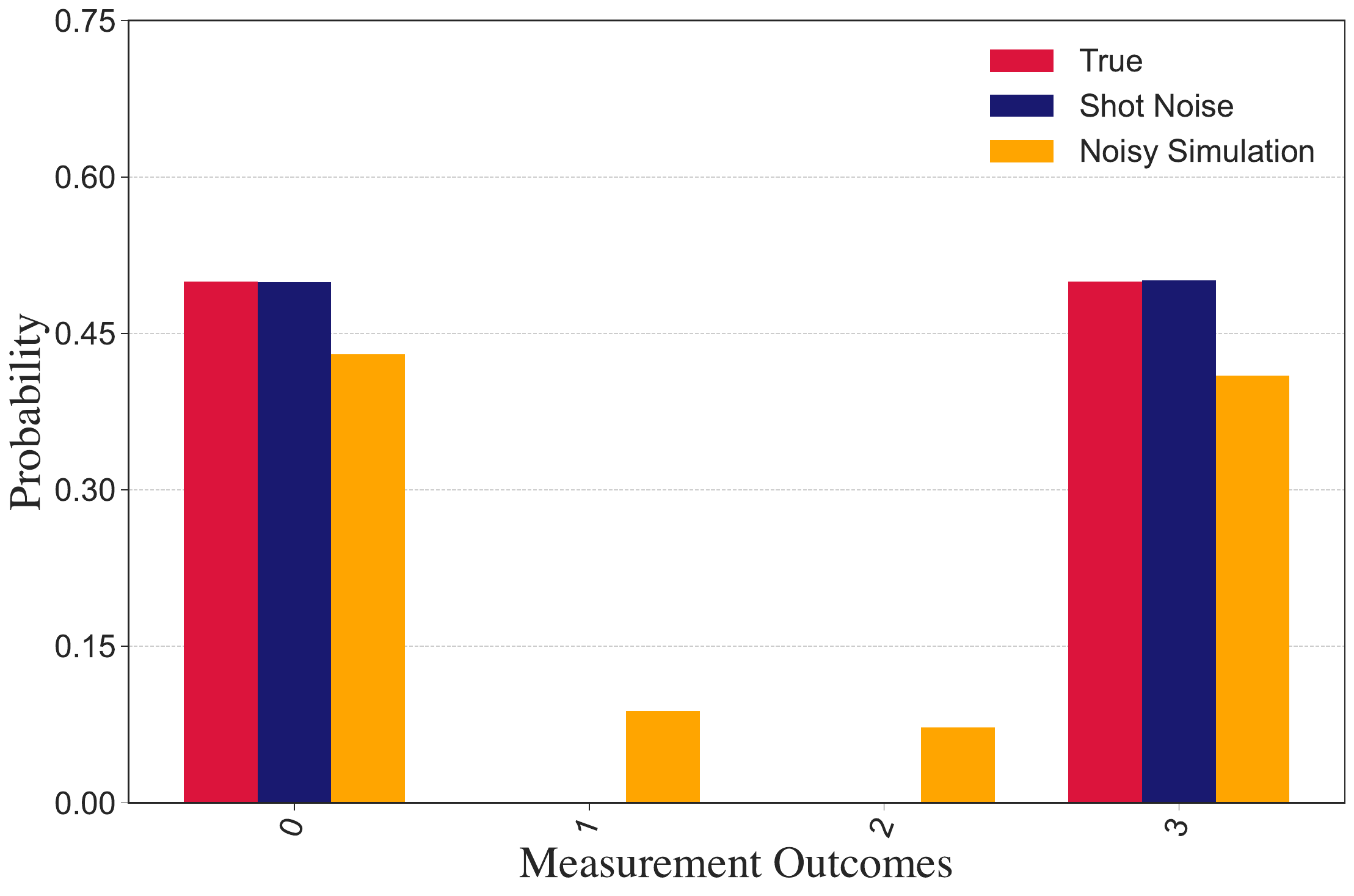}
    \end{subfigure} 
    \hfill
    \vspace{0.5cm}
    \begin{subfigure}[b]{0.49\textwidth}
    \includegraphics[width=1.0\textwidth]{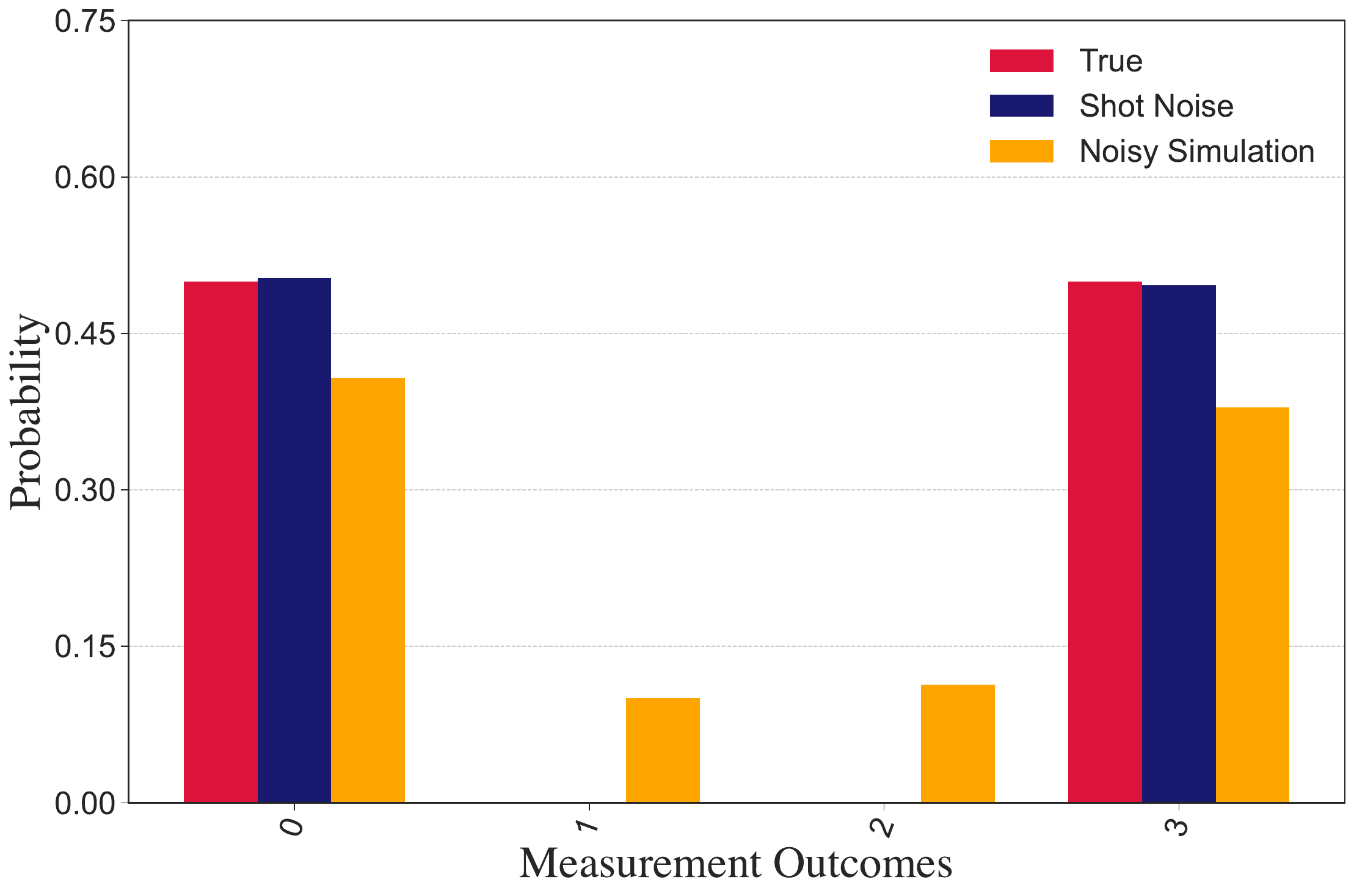}
    \end{subfigure} 
    \captionsetup{justification=centering}
    \caption{Distribution of counts obtained for the input state $\ket{\operatorname{GHZ}} = \frac{1}{\sqrt{2}}\left( \ket{000} + \ket{111}\right)$ using $s=1, 2, 3$, respectively (top left, top right, bottom), and the semi-classical encoded IQFT.}
\end{figure*}

\begin{figure*}[htbp]
    \begin{subfigure}[b]{0.49\textwidth}
    \centering
    \includegraphics[width=1.0\textwidth]{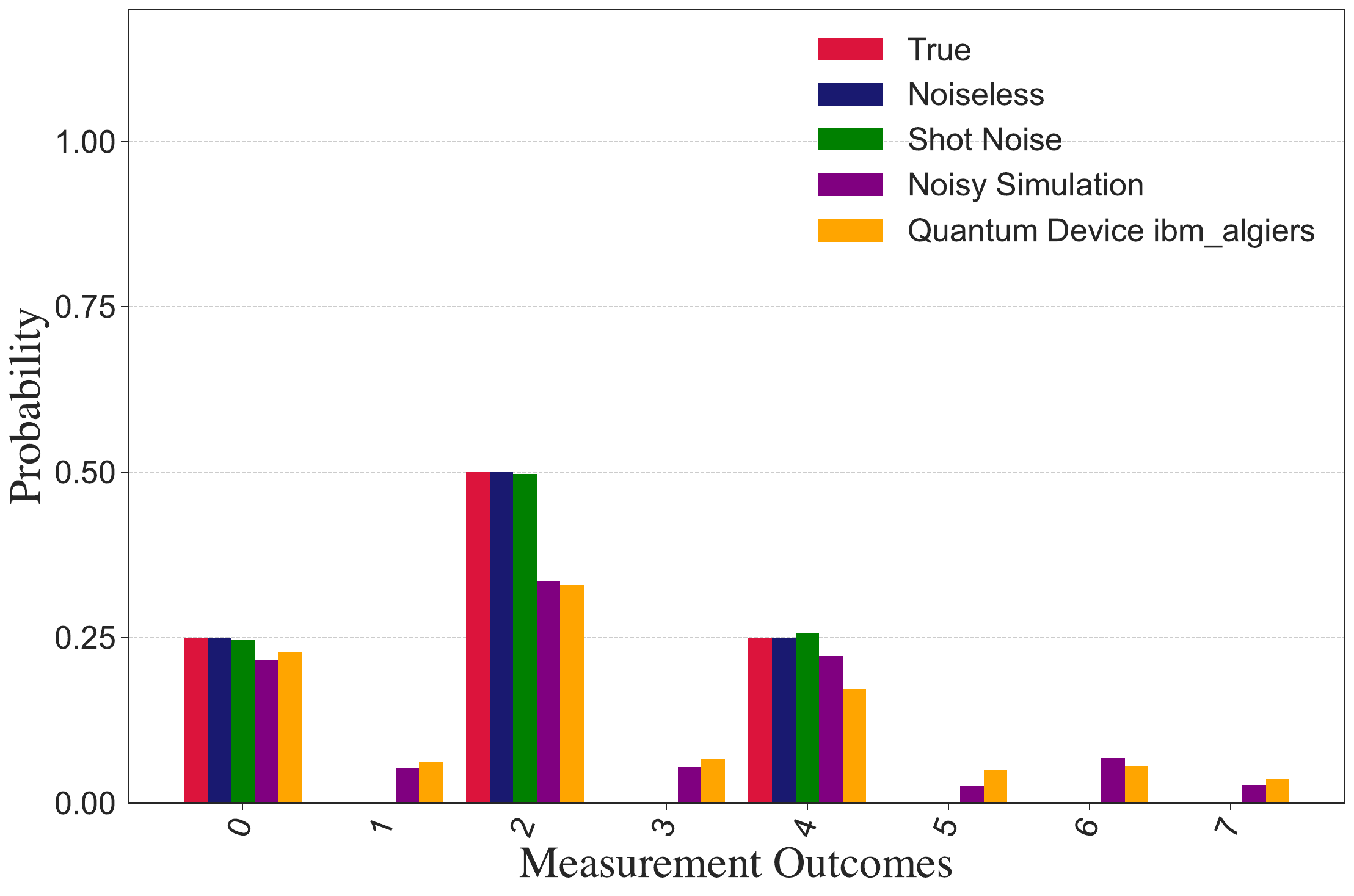}
    \end{subfigure}
    \hfill
    \begin{subfigure}[b]{0.49\textwidth}
    \centering
    \includegraphics[width=1.0\textwidth]{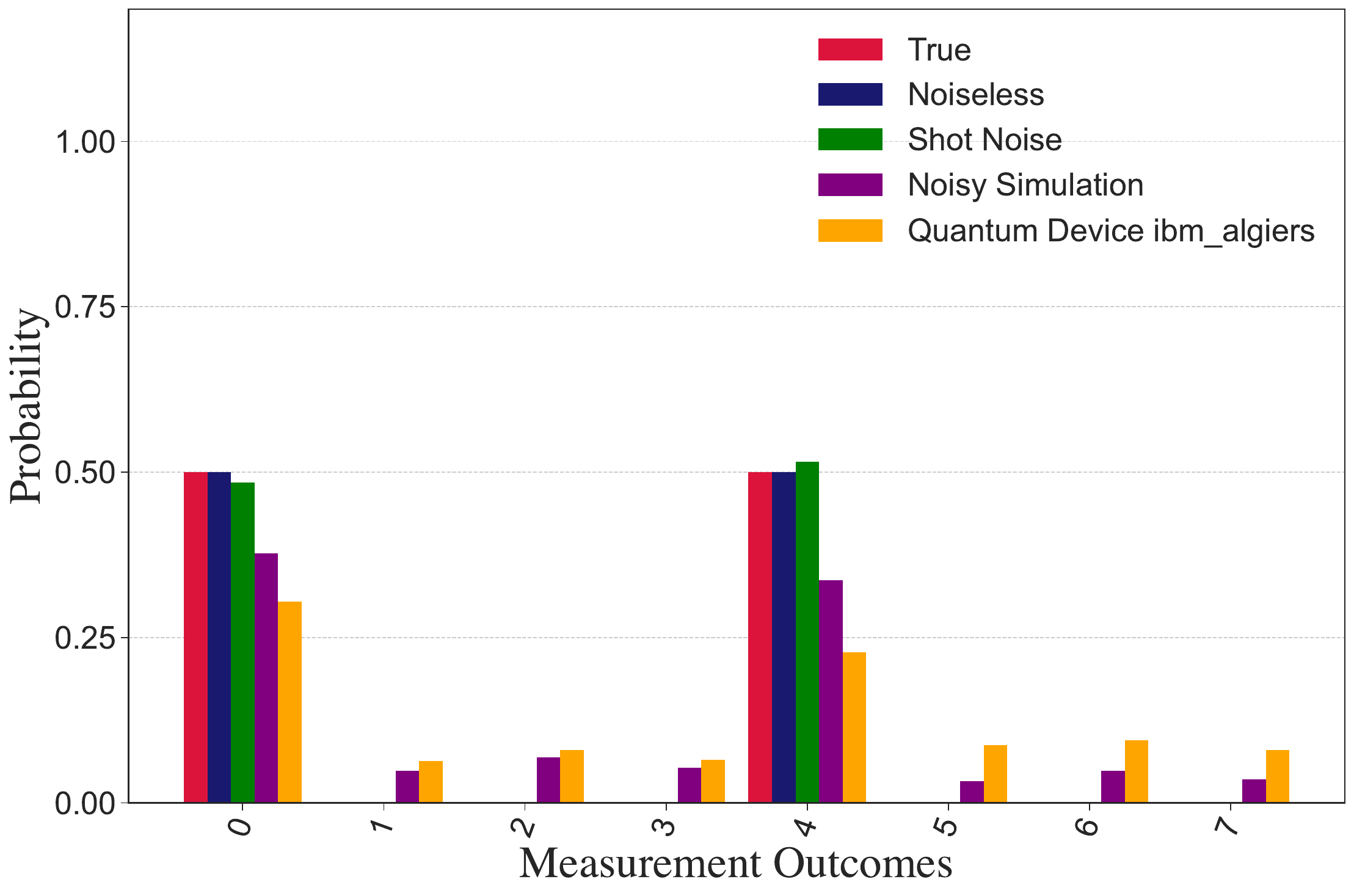}
    \end{subfigure}
    \hfill
    \begin{subfigure}[b]{0.49\textwidth}
    \centering
    \includegraphics[width=1.0\textwidth]{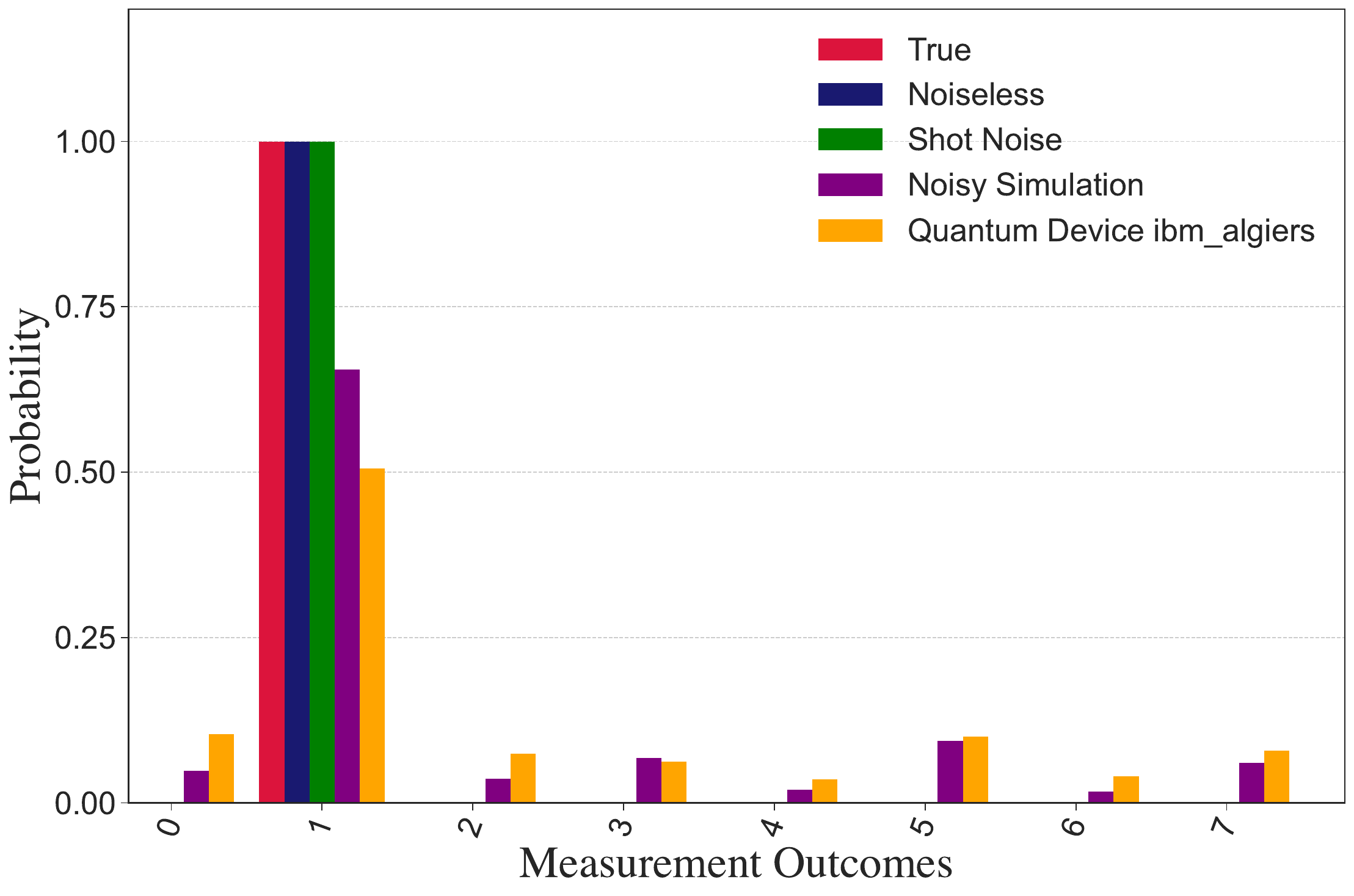}
    \end{subfigure}
    \hfill
    \begin{subfigure}[b]{0.49\textwidth}
    \centering
    \includegraphics[width=1.0\textwidth]{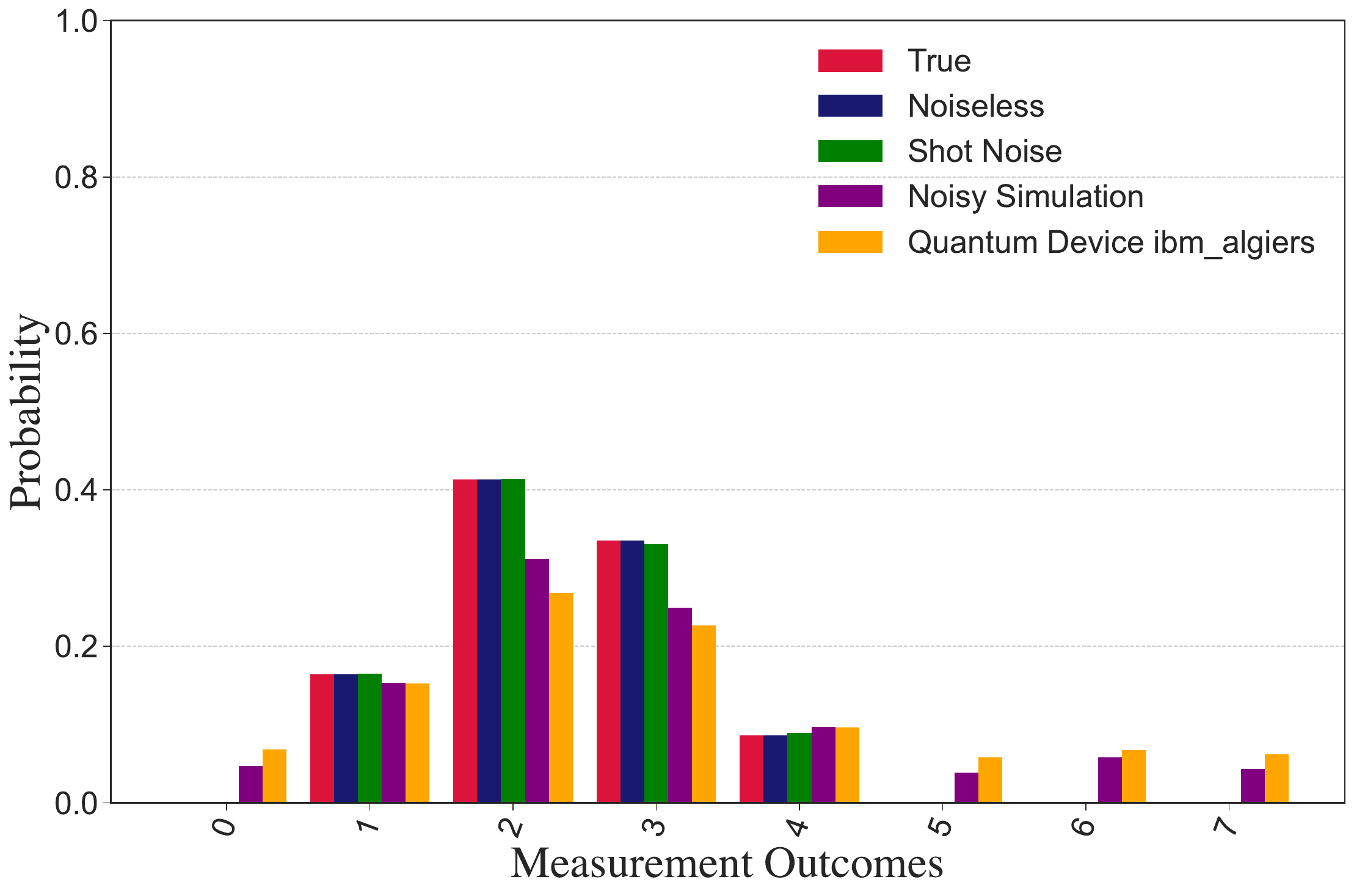}
    \end{subfigure}
    \captionsetup{justification=centering}
    \caption{Distribution of counts obtained for the input state $\ket{\psi} = \ket{\phi^+}\ket{\phi^+}$ (top left), the input state $\ket{\psi} = \frac{1}{\sqrt{2}}\left( \ket{0000} + \ket{1111}\right)$ (top right), the input state $\ket{\operatorname{W}_4} = \frac{1}{\sqrt{4}}\left( \ket{0001} + \ket{0010} + \ket{0100} + \ket{1000}\right)$ (bottom left), and a randomly generated four-qubit input state (bottom right).}
\end{figure*}

\begin{figure*}[htbp]
    \begin{subfigure}[b]{0.49\textwidth}
    \centering
    \includegraphics[width=1.0\textwidth]{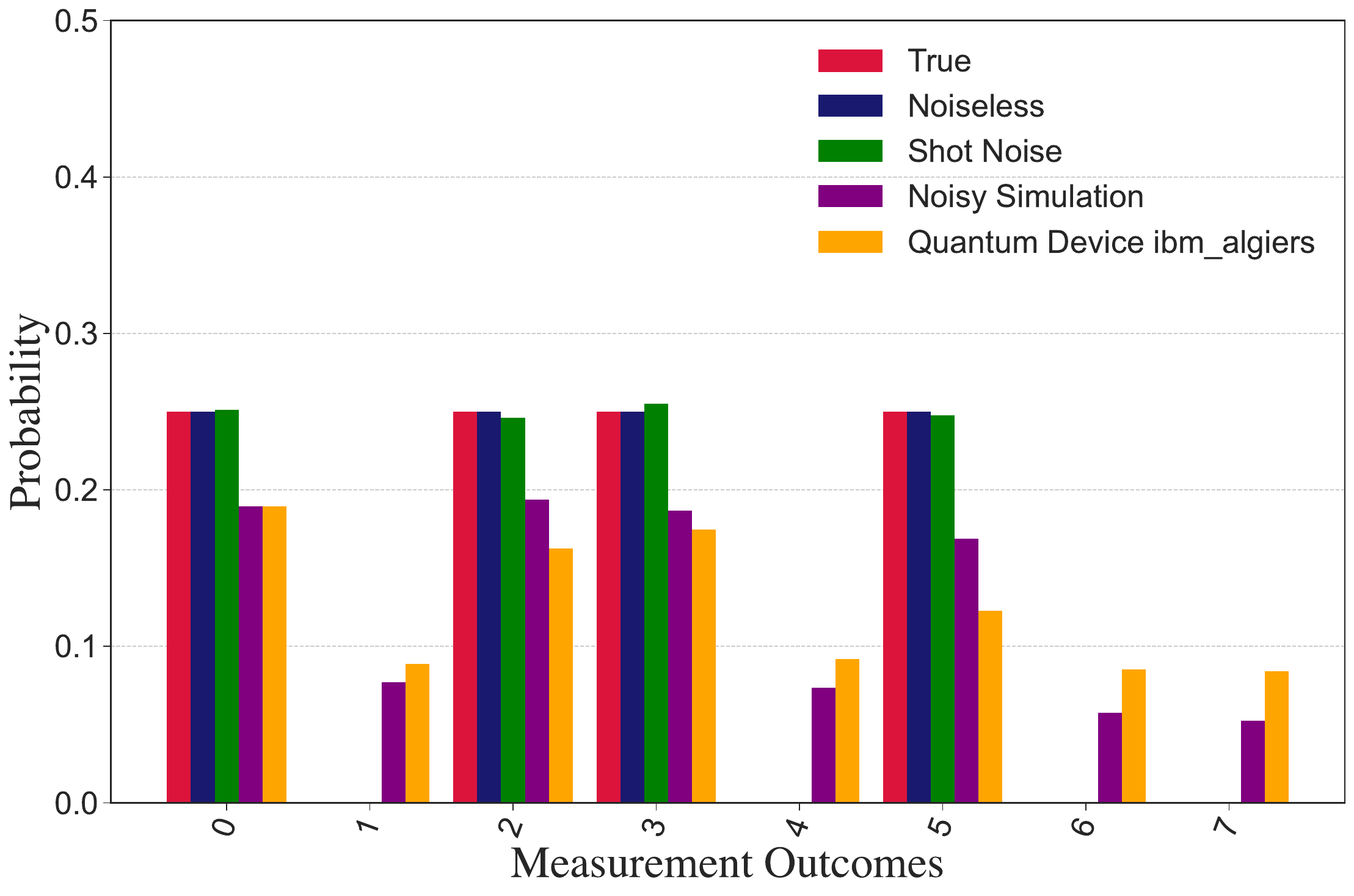}
    \end{subfigure}
    \hfill
    \begin{subfigure}[b]{0.49\textwidth}
    \centering
    \includegraphics[width=1.0\textwidth]{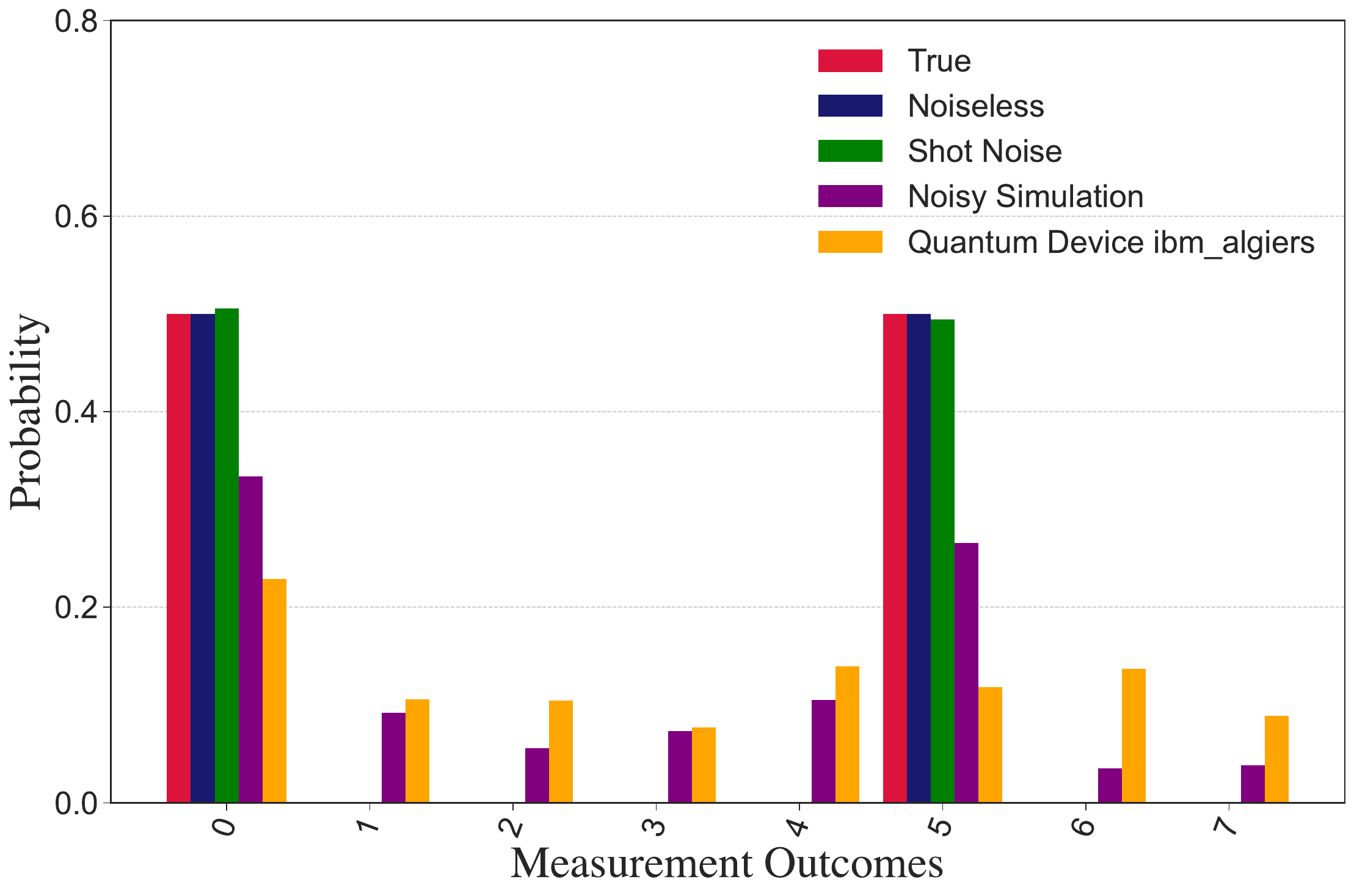}
    \end{subfigure}
    \hfill
    \begin{subfigure}[b]{0.49\textwidth}
    \centering
    \includegraphics[width=1.0\textwidth]{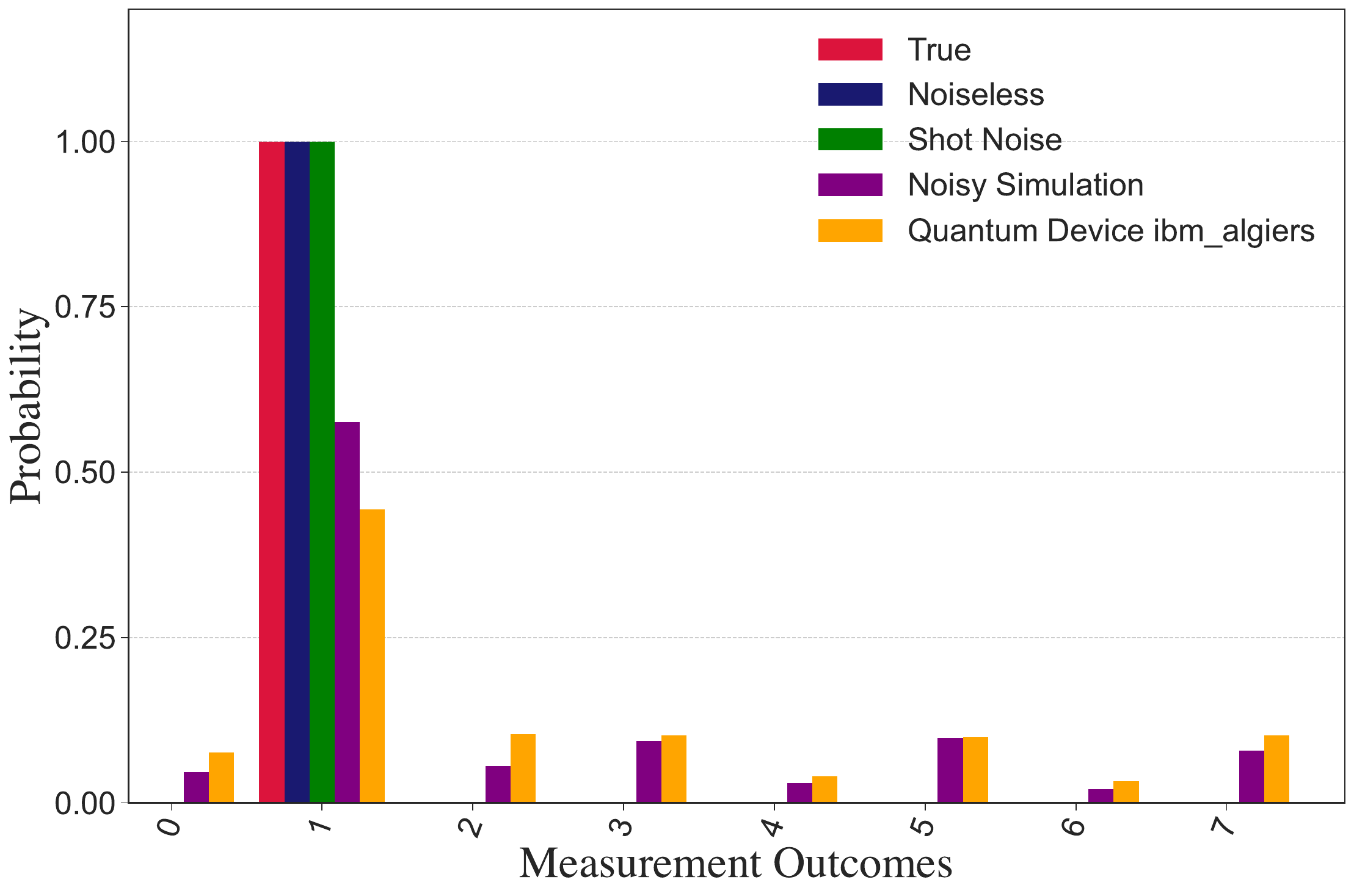}
    \end{subfigure}
    \hfill
    \begin{subfigure}[b]{0.49\textwidth}
    \centering
    \includegraphics[width=1.0\textwidth]{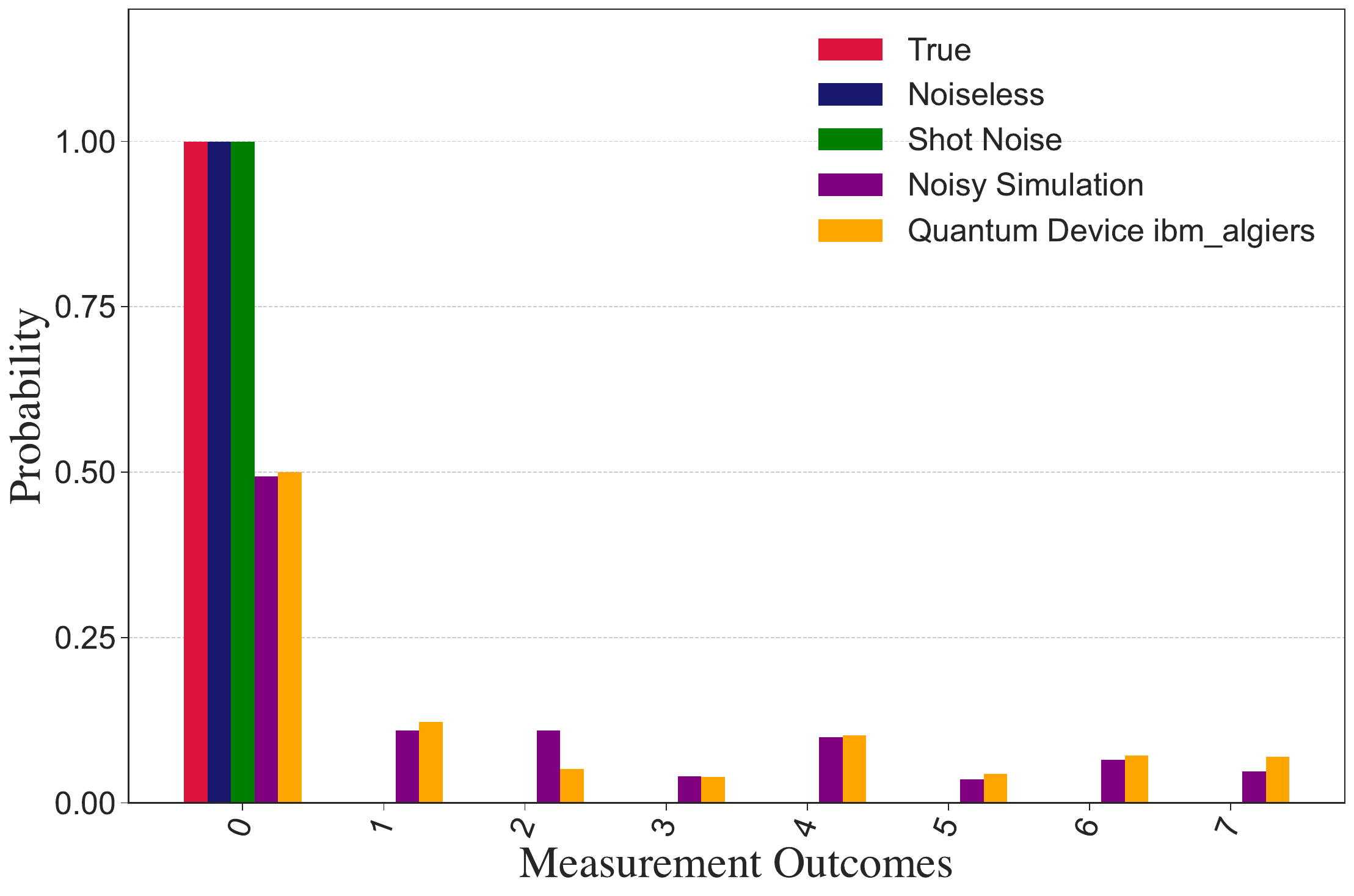}
    \end{subfigure}
    \captionsetup{justification=centering}
    \caption{Distribution of counts obtained for the input state $\ket{\psi} = \ket{\operatorname{GHZ}}\ket{\phi^+}$ (top left), the input state $\ket{\psi} = \frac{1}{\sqrt{2}}\left( \ket{00000} + \ket{11111}\right)$ (top right),  the input state $\ket{\operatorname{W}_5} = \frac{1}{\sqrt{5}}\left( \ket{00001} + \ket{00010} + \ket{00100} + \ket{01000} + \ket{10000}\right)$ (bottom left), and a randomly generated five-qubit input state (bottom right).}
\end{figure*}

\begin{figure*}[htbp]
    \begin{subfigure}[b]{0.49\textwidth}
    \centering
    \includegraphics[width=1.0\textwidth]{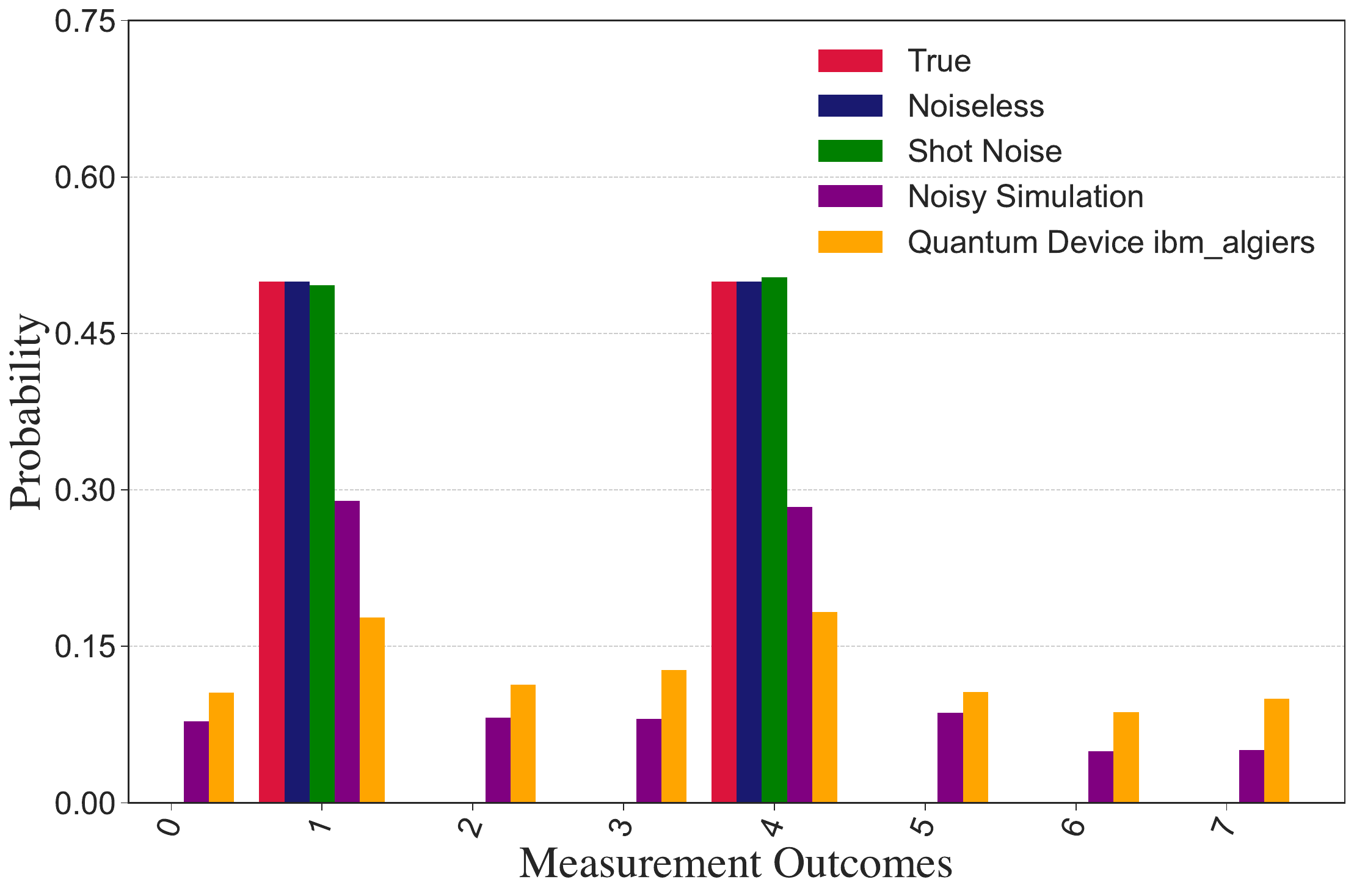}
    \end{subfigure}
    \hfill
    \begin{subfigure}[b]{0.49\textwidth}
    \centering
    \includegraphics[width=1.0\textwidth]{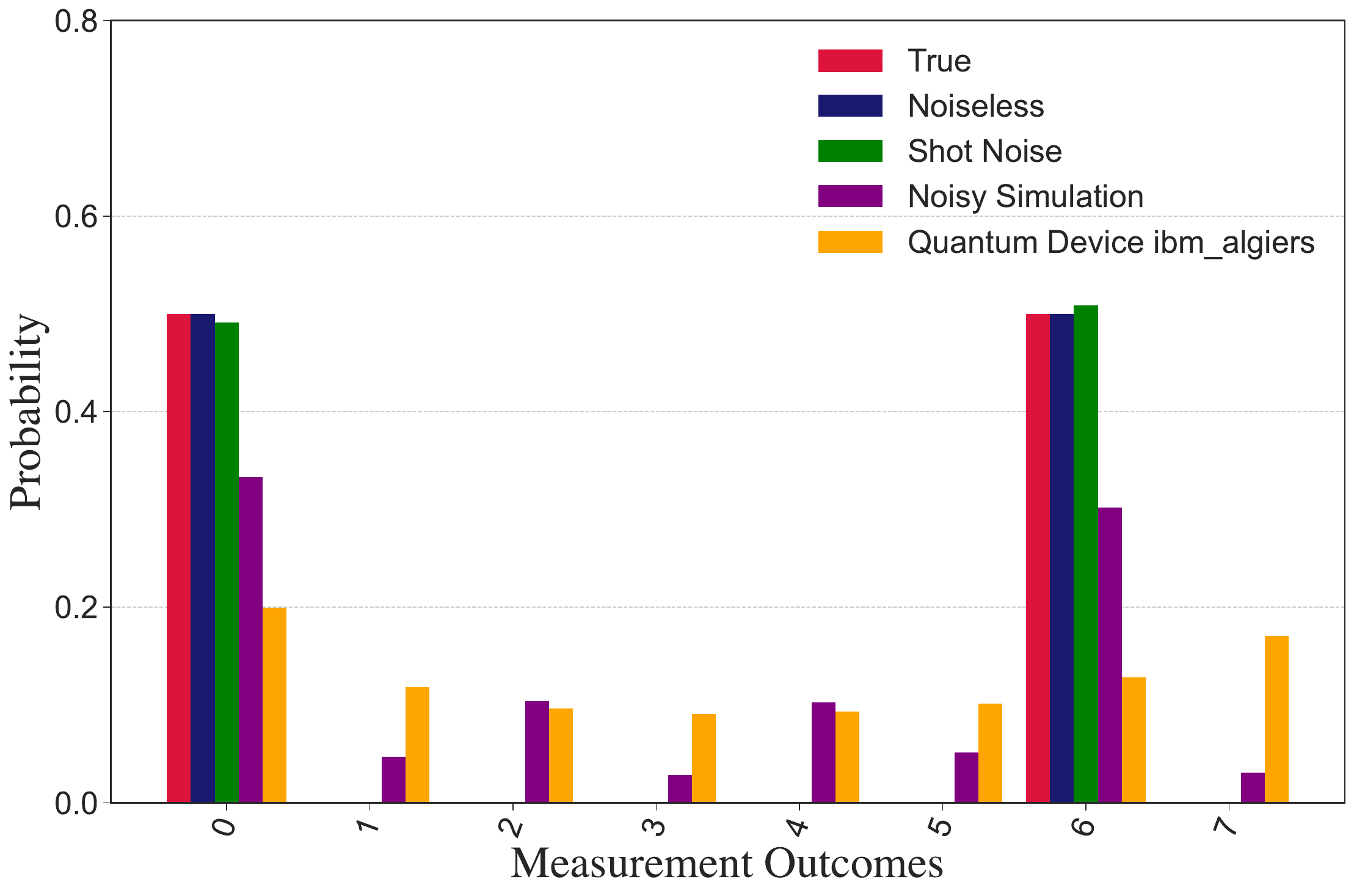}
    \end{subfigure}
    \hfill
    \begin{subfigure}[b]{0.49\textwidth}
    \centering
    \includegraphics[width=1.0\textwidth]{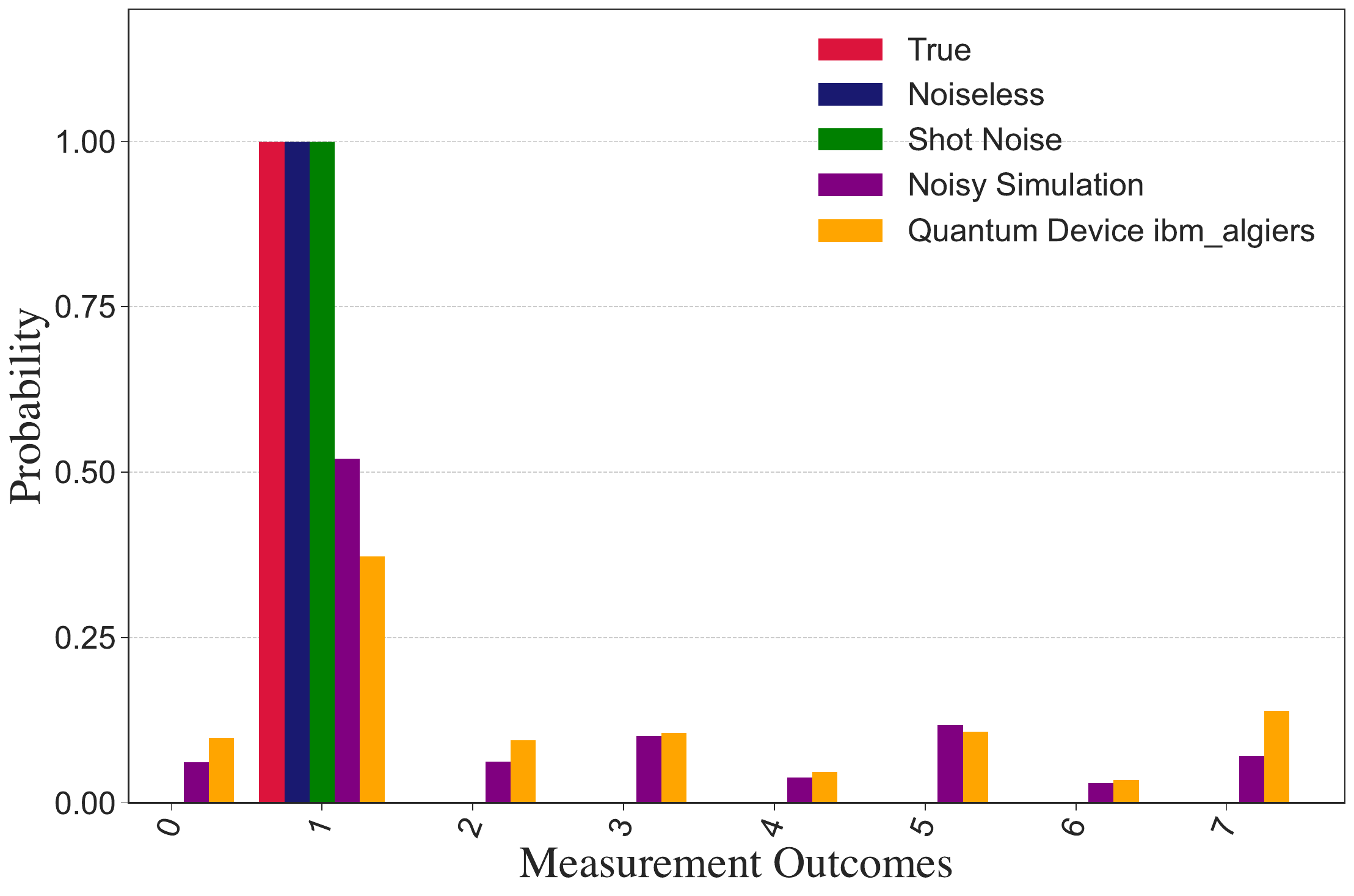}
    \end{subfigure}
    \hfill
    \begin{subfigure}[b]{0.49\textwidth}
    \centering
    \includegraphics[width=1.0\textwidth]{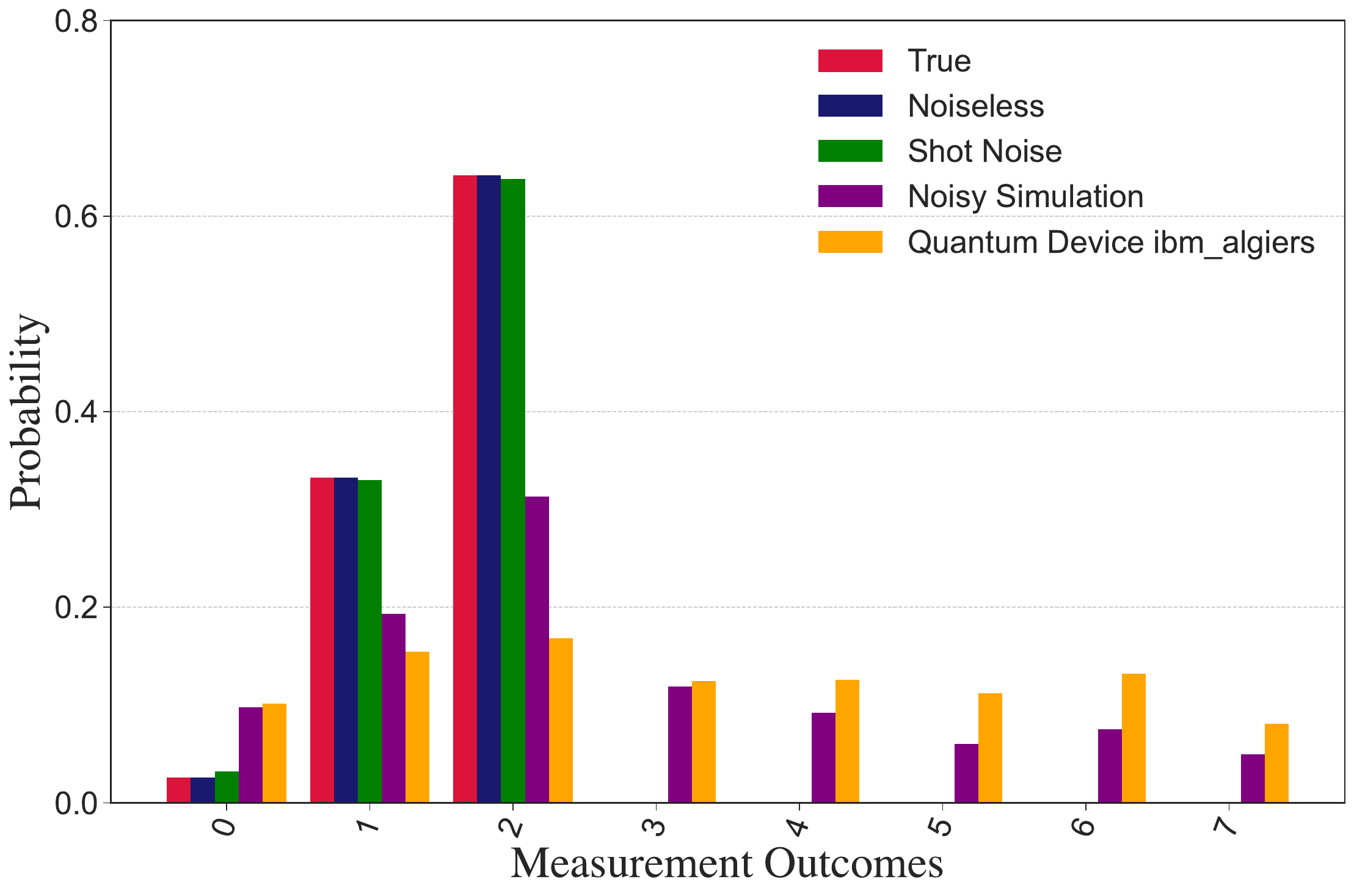}
    \end{subfigure}
    \captionsetup{justification=centering}
    \caption{Distribution of counts obtained for the input state $\ket{\psi} = \ket{\operatorname{GHZ}}\ket{\operatorname{W}}$ (top left), the input state $\ket{\psi} = \frac{1}{\sqrt{2}}\left( \ket{000000} + \ket{111111}\right)$ (top right),  the input state $\ket{\operatorname{W}_6} = \frac{1}{\sqrt{6}}( \ket{000001} + \ket{000010} + \ket{000100} + \ket{001000} + \ket{010000} + \ket{100000})$ (bottom left), and  a randomly generated six-qubit input state (bottom right).}
    \label{fig:last-example}
\end{figure*}

\section{Discussion and Conclusion}

In this work, we proposed a set of coherent quantum Hamming weight measurement circuits. The depth-optimal version has depth $\mathcal{O}(\log n)$ and requires $n$ control qubits, thus having a better depth scaling than the best previously known algorithm. If the figure of merit is the number of control qubits, the width-optimal version requires only $\mathcal{O}(\log n)$ qubits.

Future directions involve proving the optimality of the algorithm and its co-design with various available quantum computational architectures.
Another future direction is to investigate if the optimal depth scaling is retained when restricted to linear nearest neighbour architectures.

\begin{acknowledgments}

We thank Yingkai Ouyang for helpful comments about using the coherent Hamming weight measurement for quantum error detection and correction. 
SR thanks Aby Philip and Vishal Singh for helpful and illuminating discussions. SR and MMW acknowledge support from the School of Electrical and Computer Engineering at Cornell University, from the National Science Foundation under Grant No.~2315398, and from AFRL under agreement no.~FA8750-23-2-0031. MLL acknowledges support from the DoD SMART Scholarship.

This material is based on research sponsored by Air Force Research Laboratory under agreement number FA8750-23-2-0031. The U.S.~Government is authorized to reproduce and distribute reprints for Governmental purposes notwithstanding any copyright notation thereon. The views and
conclusions contained herein are those of the authors and should not be interpreted as necessarily representing the official policies or endorsements, either expressed or implied,
of Air Force Research Laboratory or the U.S.~Government.
\end{acknowledgments}

\bibliography{main}

\providecommand{\noopsort}[1]{}\providecommand{\singleletter}[1]{#1}%
\begin{thebibliography}{34}%
\makeatletter
\providecommand \@ifxundefined [1]{%
 \@ifx{#1\undefined}
}%
\providecommand \@ifnum [1]{%
 \ifnum #1\expandafter \@firstoftwo
 \else \expandafter \@secondoftwo
 \fi
}%
\providecommand \@ifx [1]{%
 \ifx #1\expandafter \@firstoftwo
 \else \expandafter \@secondoftwo
 \fi
}%
\providecommand \natexlab [1]{#1}%
\providecommand \enquote  [1]{``#1''}%
\providecommand \bibnamefont  [1]{#1}%
\providecommand \bibfnamefont [1]{#1}%
\providecommand \citenamefont [1]{#1}%
\providecommand \href@noop [0]{\@secondoftwo}%
\providecommand \href [0]{\begingroup \@sanitize@url \@href}%
\providecommand \@href[1]{\@@startlink{#1}\@@href}%
\providecommand \@@href[1]{\endgroup#1\@@endlink}%
\providecommand \@sanitize@url [0]{\catcode `\\12\catcode `\$12\catcode
  `\&12\catcode `\#12\catcode `\^12\catcode `\_12\catcode `\%12\relax}%
\providecommand \@@startlink[1]{}%
\providecommand \@@endlink[0]{}%
\providecommand \url  [0]{\begingroup\@sanitize@url \@url }%
\providecommand \@url [1]{\endgroup\@href {#1}{\urlprefix }}%
\providecommand \urlprefix  [0]{URL }%
\providecommand \Eprint [0]{\href }%
\providecommand \doibase [0]{https://doi.org/}%
\providecommand \selectlanguage [0]{\@gobble}%
\providecommand \bibinfo  [0]{\@secondoftwo}%
\providecommand \bibfield  [0]{\@secondoftwo}%
\providecommand \translation [1]{[#1]}%
\providecommand \BibitemOpen [0]{}%
\providecommand \bibitemStop [0]{}%
\providecommand \bibitemNoStop [0]{.\EOS\space}%
\providecommand \EOS [0]{\spacefactor3000\relax}%
\providecommand \BibitemShut  [1]{\csname bibitem#1\endcsname}%
\let\auto@bib@innerbib\@empty
\bibitem [{\citenamefont {Dalzell}\ \emph {et~al.}(2023)\citenamefont
  {Dalzell}, \citenamefont {McArdle}, \citenamefont {Berta}, \citenamefont
  {Bienias}, \citenamefont {Chen}, \citenamefont {Gilyén}, \citenamefont
  {Hann}, \citenamefont {Kastoryano}, \citenamefont {Khabiboulline},
  \citenamefont {Kubica}, \citenamefont {Salton}, \citenamefont {Wang},\ and\
  \citenamefont {Brandão}}]{dalzell2023quantumalgorithmssurveyapplications}%
  \BibitemOpen
  \bibfield  {author} {\bibinfo {author} {\bibfnamefont {A.~M.}\ \bibnamefont
  {Dalzell}}, \bibinfo {author} {\bibfnamefont {S.}~\bibnamefont {McArdle}},
  \bibinfo {author} {\bibfnamefont {M.}~\bibnamefont {Berta}}, \bibinfo
  {author} {\bibfnamefont {P.}~\bibnamefont {Bienias}}, \bibinfo {author}
  {\bibfnamefont {C.-F.}\ \bibnamefont {Chen}}, \bibinfo {author}
  {\bibfnamefont {A.}~\bibnamefont {Gilyén}}, \bibinfo {author} {\bibfnamefont
  {C.~T.}\ \bibnamefont {Hann}}, \bibinfo {author} {\bibfnamefont {M.~J.}\
  \bibnamefont {Kastoryano}}, \bibinfo {author} {\bibfnamefont {E.~T.}\
  \bibnamefont {Khabiboulline}}, \bibinfo {author} {\bibfnamefont
  {A.}~\bibnamefont {Kubica}}, \bibinfo {author} {\bibfnamefont
  {G.}~\bibnamefont {Salton}}, \bibinfo {author} {\bibfnamefont
  {S.}~\bibnamefont {Wang}},\ and\ \bibinfo {author} {\bibfnamefont {F.~G.
  S.~L.}\ \bibnamefont {Brandão}},\ }\href {https://arxiv.org/abs/2310.03011}
  {\bibinfo {title} {Quantum algorithms: A survey of applications and
  end-to-end complexities}} (\bibinfo {year} {2023}),\ \Eprint
  {https://arxiv.org/abs/2310.03011} {arXiv:2310.03011 [quant-ph]} \BibitemShut
  {NoStop}%
\bibitem [{\citenamefont {Cusick}\ and\ \citenamefont
  {Stanica}(2017)}]{cusick2017cryptographic}%
  \BibitemOpen
  \bibfield  {author} {\bibinfo {author} {\bibfnamefont {T.~W.}\ \bibnamefont
  {Cusick}}\ and\ \bibinfo {author} {\bibfnamefont {P.}~\bibnamefont
  {Stanica}},\ }\href@noop {} {\emph {\bibinfo {title} {Cryptographic Boolean
  Functions and Applications}}}\ (\bibinfo  {publisher} {Academic Press},\
  \bibinfo {year} {2017})\BibitemShut {NoStop}%
\bibitem [{\citenamefont {Li}\ \emph {et~al.}(2022)\citenamefont {Li},
  \citenamefont {Chen},\ and\ \citenamefont {Qu}}]{Kangquan2022generalized}%
  \BibitemOpen
  \bibfield  {author} {\bibinfo {author} {\bibfnamefont {K.}~\bibnamefont
  {Li}}, \bibinfo {author} {\bibfnamefont {H.}~\bibnamefont {Chen}},\ and\
  \bibinfo {author} {\bibfnamefont {L.}~\bibnamefont {Qu}},\ }\bibfield
  {title} {\bibinfo {title} {Generalized {H}amming weights of linear codes from
  cryptographic functions},\ }\href {https://doi.org/10.3934/amc.2022081}
  {\bibfield  {journal} {\bibinfo  {journal} {Advances in Mathematics of
  Communications}\ }\textbf {\bibinfo {volume} {16}},\ \bibinfo {pages} {859}
  (\bibinfo {year} {2022})}\BibitemShut {NoStop}%
\bibitem [{\citenamefont {Xu}\ \emph {et~al.}(2022)\citenamefont {Xu},
  \citenamefont {Fan}, \citenamefont {Mesnager}, \citenamefont {Luo},\ and\
  \citenamefont {Yan}}]{xu2023subfield}%
  \BibitemOpen
  \bibfield  {author} {\bibinfo {author} {\bibfnamefont {L.}~\bibnamefont
  {Xu}}, \bibinfo {author} {\bibfnamefont {C.}~\bibnamefont {Fan}}, \bibinfo
  {author} {\bibfnamefont {S.}~\bibnamefont {Mesnager}}, \bibinfo {author}
  {\bibfnamefont {R.}~\bibnamefont {Luo}},\ and\ \bibinfo {author}
  {\bibfnamefont {H.}~\bibnamefont {Yan}},\ }\href@noop {} {\bibinfo {title}
  {Subfield codes of several few-weight linear codes parameterized by functions
  and their consequences}} (\bibinfo {year} {2022}),\ \Eprint
  {https://arxiv.org/abs/2212.04799} {arXiv:2212.04799} \BibitemShut {NoStop}%
\bibitem [{\citenamefont {Cover}\ and\ \citenamefont {Thomas}(2006)}]{CTbook}%
  \BibitemOpen
  \bibfield  {author} {\bibinfo {author} {\bibfnamefont {T.~M.}\ \bibnamefont
  {Cover}}\ and\ \bibinfo {author} {\bibfnamefont {J.~A.}\ \bibnamefont
  {Thomas}},\ }\href@noop {} {\emph {\bibinfo {title} {Elements of Information
  Theory}}},\ Wiley Series in Telecommunications and Signal Processing\
  (\bibinfo  {publisher} {Wiley-Interscience},\ \bibinfo {address} {USA},\
  \bibinfo {year} {2006})\BibitemShut {NoStop}%
\bibitem [{\citenamefont {Stoica}\ \emph {et~al.}(2003)\citenamefont {Stoica},
  \citenamefont {Morris}, \citenamefont {Liben-Nowell}, \citenamefont {Karger},
  \citenamefont {Kaashoek}, \citenamefont {Dabek},\ and\ \citenamefont
  {Balakrishnan}}]{Stoica2003Chord}%
  \BibitemOpen
  \bibfield  {author} {\bibinfo {author} {\bibfnamefont {I.}~\bibnamefont
  {Stoica}}, \bibinfo {author} {\bibfnamefont {R.}~\bibnamefont {Morris}},
  \bibinfo {author} {\bibfnamefont {D.}~\bibnamefont {Liben-Nowell}}, \bibinfo
  {author} {\bibfnamefont {D.}~\bibnamefont {Karger}}, \bibinfo {author}
  {\bibfnamefont {M.}~\bibnamefont {Kaashoek}}, \bibinfo {author}
  {\bibfnamefont {F.}~\bibnamefont {Dabek}},\ and\ \bibinfo {author}
  {\bibfnamefont {H.}~\bibnamefont {Balakrishnan}},\ }\bibfield  {title}
  {\bibinfo {title} {Chord: a scalable peer-to-peer lookup protocol for
  internet applications},\ }\href {https://doi.org/10.1109/TNET.2002.808407}
  {\bibfield  {journal} {\bibinfo  {journal} {IEEE/ACM Transactions on
  Networking}\ }\textbf {\bibinfo {volume} {11}},\ \bibinfo {pages} {17}
  (\bibinfo {year} {2003})}\BibitemShut {NoStop}%
\bibitem [{\citenamefont {Wei}(1991)}]{wei1991generalized}%
  \BibitemOpen
  \bibfield  {author} {\bibinfo {author} {\bibfnamefont {V.~K.}\ \bibnamefont
  {Wei}},\ }\bibfield  {title} {\bibinfo {title} {Generalized {H}amming weights
  for linear codes},\ }\href {https://doi.org/10.1109/18.133259} {\bibfield
  {journal} {\bibinfo  {journal} {IEEE Transactions on Information Theory}\
  }\textbf {\bibinfo {volume} {37}},\ \bibinfo {pages} {1412} (\bibinfo {year}
  {1991})}\BibitemShut {NoStop}%
\bibitem [{\citenamefont {Wegener}(1987)}]{Wegener1987}%
  \BibitemOpen
  \bibfield  {author} {\bibinfo {author} {\bibfnamefont {I.}~\bibnamefont
  {Wegener}},\ }\bibinfo {title} {The complexity of symmetric boolean
  functions},\ in\ \href {https://doi.org/10.1007/3-540-18170-9_185} {\emph
  {\bibinfo {booktitle} {Computation Theory and Logic}}},\ \bibinfo {editor}
  {edited by\ \bibinfo {editor} {\bibfnamefont {E.}~\bibnamefont
  {B{\"o}rger}}}\ (\bibinfo  {publisher} {Springer Berlin Heidelberg},\
  \bibinfo {address} {Berlin, Heidelberg},\ \bibinfo {year} {1987})\ pp.\
  \bibinfo {pages} {433--442}\BibitemShut {NoStop}%
\bibitem [{\citenamefont {LaBorde}\ \emph {et~al.}(2023)\citenamefont
  {LaBorde}, \citenamefont {Rethinasamy},\ and\ \citenamefont
  {Wilde}}]{LaBorde2023testingsymmetry}%
  \BibitemOpen
  \bibfield  {author} {\bibinfo {author} {\bibfnamefont {M.~L.}\ \bibnamefont
  {LaBorde}}, \bibinfo {author} {\bibfnamefont {S.}~\bibnamefont
  {Rethinasamy}},\ and\ \bibinfo {author} {\bibfnamefont {M.~M.}\ \bibnamefont
  {Wilde}},\ }\bibfield  {title} {\bibinfo {title} {Testing symmetry on quantum
  computers},\ }\href {https://doi.org/10.22331/q-2023-09-25-1120} {\bibfield
  {journal} {\bibinfo  {journal} {{Quantum}}\ }\textbf {\bibinfo {volume}
  {7}},\ \bibinfo {pages} {1120} (\bibinfo {year} {2023})}\BibitemShut
  {NoStop}%
\bibitem [{\citenamefont {Warren}(2012)}]{warren2012hacker}%
  \BibitemOpen
  \bibfield  {author} {\bibinfo {author} {\bibfnamefont {H.~S.}\ \bibnamefont
  {Warren}},\ }\href@noop {} {\emph {\bibinfo {title} {Hacker's Delight}}},\
  \bibinfo {edition} {2nd}\ ed.\ (\bibinfo  {publisher} {Pearson Education},\
  \bibinfo {year} {2012})\BibitemShut {NoStop}%
\bibitem [{\citenamefont {Kaye}\ and\ \citenamefont {Mosca}(2001)}]{KM2001}%
  \BibitemOpen
  \bibfield  {author} {\bibinfo {author} {\bibfnamefont {P.}~\bibnamefont
  {Kaye}}\ and\ \bibinfo {author} {\bibfnamefont {M.}~\bibnamefont {Mosca}},\
  }\bibfield  {title} {\bibinfo {title} {Quantum networks for concentrating
  entanglement},\ }\href {https://doi.org/10.1088/0305-4470/34/35/319}
  {\bibfield  {journal} {\bibinfo  {journal} {Journal of Physics A:
  Mathematical and General}\ }\textbf {\bibinfo {volume} {34}},\ \bibinfo
  {pages} {6939} (\bibinfo {year} {2001})}\BibitemShut {NoStop}%
\bibitem [{\citenamefont {Schulte}\ \emph {et~al.}(2016)\citenamefont
  {Schulte}, \citenamefont {L\"orch}, \citenamefont {Leroux}, \citenamefont
  {Schmidt},\ and\ \citenamefont {Hammerer}}]{SLLSH16}%
  \BibitemOpen
  \bibfield  {author} {\bibinfo {author} {\bibfnamefont {M.}~\bibnamefont
  {Schulte}}, \bibinfo {author} {\bibfnamefont {N.}~\bibnamefont {L\"orch}},
  \bibinfo {author} {\bibfnamefont {I.~D.}\ \bibnamefont {Leroux}}, \bibinfo
  {author} {\bibfnamefont {P.~O.}\ \bibnamefont {Schmidt}},\ and\ \bibinfo
  {author} {\bibfnamefont {K.}~\bibnamefont {Hammerer}},\ }\bibfield  {title}
  {\bibinfo {title} {Quantum algorithmic readout in multi-ion clocks},\ }\href
  {https://doi.org/10.1103/PhysRevLett.116.013002} {\bibfield  {journal}
  {\bibinfo  {journal} {Physical Review Letters}\ }\textbf {\bibinfo {volume}
  {116}},\ \bibinfo {pages} {013002} (\bibinfo {year} {2016})}\BibitemShut
  {NoStop}%
\bibitem [{\citenamefont {Mirkamali}\ \emph {et~al.}(2018)\citenamefont
  {Mirkamali}, \citenamefont {Cory},\ and\ \citenamefont {Emerson}}]{MCE18}%
  \BibitemOpen
  \bibfield  {author} {\bibinfo {author} {\bibfnamefont {M.~S.}\ \bibnamefont
  {Mirkamali}}, \bibinfo {author} {\bibfnamefont {D.~G.}\ \bibnamefont
  {Cory}},\ and\ \bibinfo {author} {\bibfnamefont {J.}~\bibnamefont
  {Emerson}},\ }\bibfield  {title} {\bibinfo {title} {Entanglement of two
  noninteracting qubits via a mesoscopic system},\ }\href
  {https://doi.org/10.1103/PhysRevA.98.042327} {\bibfield  {journal} {\bibinfo
  {journal} {Physical Review A}\ }\textbf {\bibinfo {volume} {98}},\ \bibinfo
  {pages} {042327} (\bibinfo {year} {2018})}\BibitemShut {NoStop}%
\bibitem [{\citenamefont {Mirkamali}(2019)}]{Mirkamali2019}%
  \BibitemOpen
  \bibfield  {author} {\bibinfo {author} {\bibfnamefont {M.~S.}\ \bibnamefont
  {Mirkamali}},\ }\emph {\bibinfo {title} {Resources Needed for Entangling Two
  Qubits through an Intermediate Mesoscopic System}},\ \href
  {https://uwspace.uwaterloo.ca/handle/10012/15105} {Ph.D. thesis},\ \bibinfo
  {school} {University of Waterloo} (\bibinfo {year} {2019})\BibitemShut
  {NoStop}%
\bibitem [{\citenamefont {H{\o}yer}\ and\ \citenamefont {{\v
  S}palek}(2005)}]{HS05}%
  \BibitemOpen
  \bibfield  {author} {\bibinfo {author} {\bibfnamefont {P.}~\bibnamefont
  {H{\o}yer}}\ and\ \bibinfo {author} {\bibfnamefont {R.}~\bibnamefont {{\v
  S}palek}},\ }\bibfield  {title} {\bibinfo {title} {Quantum fan-out is
  powerful},\ }\href {https://doi.org/10.4086/toc.2005.v001a005} {\bibfield
  {journal} {\bibinfo  {journal} {Theory of Computing}\ }\textbf {\bibinfo
  {volume} {1}},\ \bibinfo {pages} {81} (\bibinfo {year} {2005})}\BibitemShut
  {NoStop}%
\bibitem [{\citenamefont {Pham}\ and\ \citenamefont
  {Svore}(2013)}]{2D_Factoring}%
  \BibitemOpen
  \bibfield  {author} {\bibinfo {author} {\bibfnamefont {P.}~\bibnamefont
  {Pham}}\ and\ \bibinfo {author} {\bibfnamefont {K.~M.}\ \bibnamefont
  {Svore}},\ }\bibfield  {title} {\bibinfo {title} {A {2D} nearest-neighbor
  quantum architecture for factoring in polylogarithmic depth},\ }\href
  {https://doi.org/10.26421/QIC13.11-12-3} {\bibfield  {journal} {\bibinfo
  {journal} {Quantum Information and Computation}\ }\textbf {\bibinfo {volume}
  {13}},\ \bibinfo {pages} {937–962} (\bibinfo {year} {2013})}\BibitemShut
  {NoStop}%
\bibitem [{\citenamefont {Draper}(2000)}]{draper2000}%
  \BibitemOpen
  \bibfield  {author} {\bibinfo {author} {\bibfnamefont {T.~G.}\ \bibnamefont
  {Draper}},\ }\href {https://arxiv.org/abs/quant-ph/0008033} {\bibinfo {title}
  {Addition on a quantum computer}} (\bibinfo {year} {2000}),\ \Eprint
  {https://arxiv.org/abs/quant-ph/0008033} {arXiv:quant-ph/0008033 [quant-ph]}
  \BibitemShut {NoStop}%
\bibitem [{\citenamefont {Bennett}\ \emph {et~al.}(1996)\citenamefont
  {Bennett}, \citenamefont {Bernstein}, \citenamefont {Popescu},\ and\
  \citenamefont {Schumacher}}]{BBPS96}%
  \BibitemOpen
  \bibfield  {author} {\bibinfo {author} {\bibfnamefont {C.~H.}\ \bibnamefont
  {Bennett}}, \bibinfo {author} {\bibfnamefont {H.~J.}\ \bibnamefont
  {Bernstein}}, \bibinfo {author} {\bibfnamefont {S.}~\bibnamefont {Popescu}},\
  and\ \bibinfo {author} {\bibfnamefont {B.}~\bibnamefont {Schumacher}},\
  }\bibfield  {title} {\bibinfo {title} {Concentrating partial entanglement by
  local operations},\ }\href {https://doi.org/10.1103/PhysRevA.53.2046}
  {\bibfield  {journal} {\bibinfo  {journal} {Physical Review A}\ }\textbf
  {\bibinfo {volume} {53}},\ \bibinfo {pages} {2046} (\bibinfo {year}
  {1996})}\BibitemShut {NoStop}%
\bibitem [{\citenamefont {Wilde}(2017)}]{Wilde_2017}%
  \BibitemOpen
  \bibfield  {author} {\bibinfo {author} {\bibfnamefont {M.~M.}\ \bibnamefont
  {Wilde}},\ }\href {https://doi.org/10.1017/9781316809976} {\emph {\bibinfo
  {title} {Quantum Information Theory}}},\ \bibinfo {edition} {2nd}\ ed.\
  (\bibinfo  {publisher} {Cambridge University Press},\ \bibinfo {year}
  {2017})\BibitemShut {NoStop}%
\bibitem [{\citenamefont {Salek}\ and\ \citenamefont {Winter}(2022)}]{SW22}%
  \BibitemOpen
  \bibfield  {author} {\bibinfo {author} {\bibfnamefont {F.}~\bibnamefont
  {Salek}}\ and\ \bibinfo {author} {\bibfnamefont {A.}~\bibnamefont {Winter}},\
  }\bibfield  {title} {\bibinfo {title} {Multi-user distillation of common
  randomness and entanglement from quantum states},\ }\href
  {https://doi.org/10.1109/TIT.2021.3124965} {\bibfield  {journal} {\bibinfo
  {journal} {IEEE Transactions on Information Theory}\ }\textbf {\bibinfo
  {volume} {68}},\ \bibinfo {pages} {976} (\bibinfo {year} {2022})}\BibitemShut
  {NoStop}%
\bibitem [{\citenamefont {Winter}\ and\ \citenamefont
  {Yang}(2016)}]{PhysRevLett.116.120404}%
  \BibitemOpen
  \bibfield  {author} {\bibinfo {author} {\bibfnamefont {A.}~\bibnamefont
  {Winter}}\ and\ \bibinfo {author} {\bibfnamefont {D.}~\bibnamefont {Yang}},\
  }\bibfield  {title} {\bibinfo {title} {Operational resource theory of
  coherence},\ }\href {https://doi.org/10.1103/PhysRevLett.116.120404}
  {\bibfield  {journal} {\bibinfo  {journal} {Physical Review Letters}\
  }\textbf {\bibinfo {volume} {116}},\ \bibinfo {pages} {120404} (\bibinfo
  {year} {2016})}\BibitemShut {NoStop}%
\bibitem [{\citenamefont {Shih}\ \emph {et~al.}(2010)\citenamefont {Shih},
  \citenamefont {Hsieh},\ and\ \citenamefont {Wei}}]{SHIH2010273}%
  \BibitemOpen
  \bibfield  {author} {\bibinfo {author} {\bibfnamefont {Y.-C.}\ \bibnamefont
  {Shih}}, \bibinfo {author} {\bibfnamefont {M.-H.}\ \bibnamefont {Hsieh}},\
  and\ \bibinfo {author} {\bibfnamefont {H.-Y.}\ \bibnamefont {Wei}},\
  }\bibfield  {title} {\bibinfo {title} {Multicasting homogeneous and
  heterogeneous quantum states in quantum networks},\ }\href
  {https://doi.org/https://doi.org/10.1016/j.nancom.2010.10.003} {\bibfield
  {journal} {\bibinfo  {journal} {Nano Communication Networks}\ }\textbf
  {\bibinfo {volume} {1}},\ \bibinfo {pages} {273} (\bibinfo {year}
  {2010})}\BibitemShut {NoStop}%
\bibitem [{\citenamefont {Plesch}\ and\ \citenamefont
  {Bu\ifmmode~\check{z}\else \v{z}\fi{}ek}(2010)}]{PhysRevA.81.032317}%
  \BibitemOpen
  \bibfield  {author} {\bibinfo {author} {\bibfnamefont {M.}~\bibnamefont
  {Plesch}}\ and\ \bibinfo {author} {\bibfnamefont {V.}~\bibnamefont
  {Bu\ifmmode~\check{z}\else \v{z}\fi{}ek}},\ }\bibfield  {title} {\bibinfo
  {title} {Efficient compression of quantum information},\ }\href
  {https://doi.org/10.1103/PhysRevA.81.032317} {\bibfield  {journal} {\bibinfo
  {journal} {Physical Review A}\ }\textbf {\bibinfo {volume} {81}},\ \bibinfo
  {pages} {032317} (\bibinfo {year} {2010})}\BibitemShut {NoStop}%
\bibitem [{\citenamefont {Ouyang}(2021)}]{OUY21}%
  \BibitemOpen
  \bibfield  {author} {\bibinfo {author} {\bibfnamefont {Y.}~\bibnamefont
  {Ouyang}},\ }\bibfield  {title} {\bibinfo {title} {Avoiding coherent errors
  with rotated concatenated stabilizer codes},\ }\href
  {https://doi.org/10.1038/s41534-021-00429-8} {\bibfield  {journal} {\bibinfo
  {journal} {npj Quantum Information}\ }\textbf {\bibinfo {volume} {7}},\
  \bibinfo {pages} {87} (\bibinfo {year} {2021})}\BibitemShut {NoStop}%
\bibitem [{\citenamefont {C\'orcoles}\ \emph {et~al.}(2021)\citenamefont
  {C\'orcoles}, \citenamefont {Takita}, \citenamefont {Inoue}, \citenamefont
  {Lekuch}, \citenamefont {Minev}, \citenamefont {Chow},\ and\ \citenamefont
  {Gambetta}}]{PhysRevLett.127.100501}%
  \BibitemOpen
  \bibfield  {author} {\bibinfo {author} {\bibfnamefont {A.~D.}\ \bibnamefont
  {C\'orcoles}}, \bibinfo {author} {\bibfnamefont {M.}~\bibnamefont {Takita}},
  \bibinfo {author} {\bibfnamefont {K.}~\bibnamefont {Inoue}}, \bibinfo
  {author} {\bibfnamefont {S.}~\bibnamefont {Lekuch}}, \bibinfo {author}
  {\bibfnamefont {Z.~K.}\ \bibnamefont {Minev}}, \bibinfo {author}
  {\bibfnamefont {J.~M.}\ \bibnamefont {Chow}},\ and\ \bibinfo {author}
  {\bibfnamefont {J.~M.}\ \bibnamefont {Gambetta}},\ }\bibfield  {title}
  {\bibinfo {title} {Exploiting dynamic quantum circuits in a quantum algorithm
  with superconducting qubits},\ }\href
  {https://doi.org/10.1103/PhysRevLett.127.100501} {\bibfield  {journal}
  {\bibinfo  {journal} {Physical Review Letters}\ }\textbf {\bibinfo {volume}
  {127}},\ \bibinfo {pages} {100501} (\bibinfo {year} {2021})}\BibitemShut
  {NoStop}%
\bibitem [{\citenamefont {DeCross}\ \emph {et~al.}(2023)\citenamefont
  {DeCross}, \citenamefont {Chertkov}, \citenamefont {Kohagen},\ and\
  \citenamefont {Foss-Feig}}]{PhysRevX.13.041057}%
  \BibitemOpen
  \bibfield  {author} {\bibinfo {author} {\bibfnamefont {M.}~\bibnamefont
  {DeCross}}, \bibinfo {author} {\bibfnamefont {E.}~\bibnamefont {Chertkov}},
  \bibinfo {author} {\bibfnamefont {M.}~\bibnamefont {Kohagen}},\ and\ \bibinfo
  {author} {\bibfnamefont {M.}~\bibnamefont {Foss-Feig}},\ }\bibfield  {title}
  {\bibinfo {title} {Qubit-reuse compilation with mid-circuit measurement and
  reset},\ }\href {https://doi.org/10.1103/PhysRevX.13.041057} {\bibfield
  {journal} {\bibinfo  {journal} {Physical Review X}\ }\textbf {\bibinfo
  {volume} {13}},\ \bibinfo {pages} {041057} (\bibinfo {year}
  {2023})}\BibitemShut {NoStop}%
\bibitem [{\citenamefont {Piroli}\ \emph {et~al.}(2024)\citenamefont {Piroli},
  \citenamefont {Styliaris},\ and\ \citenamefont
  {Cirac}}]{piroli2024approximating}%
  \BibitemOpen
  \bibfield  {author} {\bibinfo {author} {\bibfnamefont {L.}~\bibnamefont
  {Piroli}}, \bibinfo {author} {\bibfnamefont {G.}~\bibnamefont {Styliaris}},\
  and\ \bibinfo {author} {\bibfnamefont {J.~I.}\ \bibnamefont {Cirac}},\
  }\href@noop {} {\bibinfo {title} {Approximating many-body quantum states with
  quantum circuits and measurements}} (\bibinfo {year} {2024}),\ \Eprint
  {https://arxiv.org/abs/2403.07604} {arXiv:2403.07604 [quant-ph]} \BibitemShut
  {NoStop}%
\bibitem [{\citenamefont {Zi}\ \emph {et~al.}(2024)\citenamefont {Zi},
  \citenamefont {Nie},\ and\ \citenamefont {Sun}}]{zi2024shallow}%
  \BibitemOpen
  \bibfield  {author} {\bibinfo {author} {\bibfnamefont {W.}~\bibnamefont
  {Zi}}, \bibinfo {author} {\bibfnamefont {J.}~\bibnamefont {Nie}},\ and\
  \bibinfo {author} {\bibfnamefont {X.}~\bibnamefont {Sun}},\ }\href@noop {}
  {\bibinfo {title} {Shallow quantum circuit implementation of symmetric
  functions with limited ancillary qubits}} (\bibinfo {year} {2024}),\ \Eprint
  {https://arxiv.org/abs/2404.06052} {arXiv:2404.06052 [quant-ph]} \BibitemShut
  {NoStop}%
\bibitem [{\citenamefont {Tomamichel}(2015)}]{tomamichel2015quantum}%
  \BibitemOpen
  \bibfield  {author} {\bibinfo {author} {\bibfnamefont {M.}~\bibnamefont
  {Tomamichel}},\ }\href {https://doi.org/10.1007/978-3-319-21891-5} {\emph
  {\bibinfo {title} {Quantum Information Processing with Finite Resources:
  Mathematical Foundations}}},\ Vol.~\bibinfo {volume} {5}\ (\bibinfo
  {publisher} {Springer},\ \bibinfo {year} {2015})\BibitemShut {NoStop}%
\bibitem [{\citenamefont {Quek}\ \emph {et~al.}(2024)\citenamefont {Quek},
  \citenamefont {Kaur},\ and\ \citenamefont {Wilde}}]{QKW24}%
  \BibitemOpen
  \bibfield  {author} {\bibinfo {author} {\bibfnamefont {Y.}~\bibnamefont
  {Quek}}, \bibinfo {author} {\bibfnamefont {E.}~\bibnamefont {Kaur}},\ and\
  \bibinfo {author} {\bibfnamefont {M.~M.}\ \bibnamefont {Wilde}},\ }\bibfield
  {title} {\bibinfo {title} {Multivariate trace estimation in constant quantum
  depth},\ }\href {https://doi.org/10.22331/q-2024-01-10-1220} {\bibfield
  {journal} {\bibinfo  {journal} {{Quantum}}\ }\textbf {\bibinfo {volume}
  {8}},\ \bibinfo {pages} {1220} (\bibinfo {year} {2024})}\BibitemShut
  {NoStop}%
\bibitem [{\citenamefont {Shor}(1996)}]{S96}%
  \BibitemOpen
  \bibfield  {author} {\bibinfo {author} {\bibfnamefont {P.~W.}\ \bibnamefont
  {Shor}},\ }\bibfield  {title} {\bibinfo {title} {Fault-tolerant quantum
  computation},\ }in\ \href {https://doi.org/10.1109/SFCS.1996.548464} {\emph
  {\bibinfo {booktitle} {Proceedings of the 37th Annual Symposium on
  Foundations of Computer Science}}},\ \bibinfo {series and number} {FOCS '96}\
  (\bibinfo  {publisher} {IEEE Computer Society},\ \bibinfo {address} {USA},\
  \bibinfo {year} {1996})\ p.~\bibinfo {pages} {56},\ \bibinfo {note}
  {arXiv:quant-ph/9605011}\BibitemShut {NoStop}%
\bibitem [{\citenamefont {Gottesman}(2010)}]{G10}%
  \BibitemOpen
  \bibfield  {author} {\bibinfo {author} {\bibfnamefont {D.}~\bibnamefont
  {Gottesman}},\ }\bibfield  {title} {\bibinfo {title} {An introduction to
  quantum error correction and fault-tolerant quantum computation},\
  }\href@noop {} {\bibfield  {journal} {\bibinfo  {journal} {Quantum
  Information Science and Its Contributions to Mathematics, Proceedings of
  Symposia in Applied Mathematics}\ }\textbf {\bibinfo {volume} {68}},\
  \bibinfo {pages} {13} (\bibinfo {year} {2010})},\ \Eprint
  {https://arxiv.org/abs/arXiv:0904.2557} {arXiv:0904.2557} \BibitemShut
  {NoStop}%
\bibitem [{\citenamefont {Griffiths}\ and\ \citenamefont
  {Niu}(1996)}]{PhysRevLett.76.3228}%
  \BibitemOpen
  \bibfield  {author} {\bibinfo {author} {\bibfnamefont {R.~B.}\ \bibnamefont
  {Griffiths}}\ and\ \bibinfo {author} {\bibfnamefont {C.-S.}\ \bibnamefont
  {Niu}},\ }\bibfield  {title} {\bibinfo {title} {Semiclassical {F}ourier
  transform for quantum computation},\ }\href
  {https://doi.org/10.1103/PhysRevLett.76.3228} {\bibfield  {journal} {\bibinfo
   {journal} {Physical Review Letters}\ }\textbf {\bibinfo {volume} {76}},\
  \bibinfo {pages} {3228} (\bibinfo {year} {1996})}\BibitemShut {NoStop}%
\bibitem [{\citenamefont {Smolin}\ \emph {et~al.}(2013)\citenamefont {Smolin},
  \citenamefont {Smith},\ and\ \citenamefont {Vargo}}]{Smolin2013}%
  \BibitemOpen
  \bibfield  {author} {\bibinfo {author} {\bibfnamefont {J.~A.}\ \bibnamefont
  {Smolin}}, \bibinfo {author} {\bibfnamefont {G.}~\bibnamefont {Smith}},\ and\
  \bibinfo {author} {\bibfnamefont {A.}~\bibnamefont {Vargo}},\ }\bibfield
  {title} {\bibinfo {title} {Oversimplifying quantum factoring},\ }\href
  {https://doi.org/10.1038/nature12290} {\bibfield  {journal} {\bibinfo
  {journal} {Nature}\ }\textbf {\bibinfo {volume} {499}},\ \bibinfo {pages}
  {163} (\bibinfo {year} {2013})},\ \Eprint {https://arxiv.org/abs/1301.7007}
  {arXiv:1301.7007} \BibitemShut {NoStop}%
\end{thebibliography}%

\end{document}